%% file: main.tex
\documentclass[aps,
prx,
twocolumn,superscriptaddress,nofootinbib,longbibliography,floatfix]{revtex4-2}

\usepackage[T1]{fontenc}
\usepackage{amsmath,amssymb,amsthm,bm,mathtools}
\usepackage{graphicx}
\usepackage{microtype}
\usepackage{booktabs}
\usepackage{braket}
\usepackage[colorlinks=true,linkcolor=blue,citecolor=blue,urlcolor=blue]{hyperref}

\newtheorem{theorem}{Theorem}

\newtheorem{remark}{Remark}

\newcommand{\Uc}{U_c}
\newcommand{\avg}[1]{\langle #1\rangle}
\newcommand{\Jpf}{\mathcal J_{\rm PF}}
\newcommand{\ekin}{\varepsilon_{\rm kin}^{(0)}}

\newcommand{\nmode}{\nu_{\rm m}}
\newcommand{\Rp}{\mathcal R_p}
\newcommand{\rb}{r_{\rm b}}

\newcommand{\Tr}{\operatorname{Tr}}

\newcommand{\mytitle}{
When Can a Cavity Move a Mott Transition? A Spectral-Density Criterion within Gutzwiller Theory
}

\begin{document}

\title{\mytitle}

\author{Nikhil Vamsodharakan Seshadri}
\affiliation{Los Alamos National Laboratory, Los Alamos, New Mexico 87545, USA}

\author{Yu Zhang}
\email{zhy@lanl.gov}
\affiliation{Los Alamos National Laboratory, Los Alamos, New Mexico 87545, USA}

\date{\today}

\begin{abstract}
Can vacuum electromagnetic fluctuations shift a bulk Mott transition? Within the Gutzwiller variational method, we derive a criterion that separates collective spectroscopic hybridization from thermodynamic phase control. We show that a Mott transition shifts only when the electromagnetic environment supplies finite thermodynamic spectral weight with bond-scale variation. A joint frequency--spatial Pauli--Fierz density gives the leading shift. Surface phonon polaritons yield a $d^{-3}$-to-$d^{-5}$ crossover, while finite-coordination variational Monte Carlo supports the predicted critical coefficient and $M/N$ scaling.
\end{abstract}

\maketitle

Correlation-driven metal--insulator transitions or phase transitions at large are central problems of condensed-matter physics~\cite{Mott1968, Hubbard1963, Imada1998, Georges1996, Kotliar2006}. Quantized electromagnetic environments offer a complementary route to modify collective states and microscopic interactions~\cite{Schlawin2022, Hubener2021, Sentef2018, Curtis2019, Kiffner2019, Mazza2019, Weight:2023aa, Bauman:2025aa, Keren:2026aa}, yet whether vacuum fields can alter an equilibrium bulk phase remains unsettled. A normalized bright mode can exhibit a finite collective Rabi splitting in the thermodynamic limit even though its contribution to the energy density vanishes. Recent quantum Monte Carlo (QMC) found that a properly normalized single mode is irrelevant at the honeycomb-lattice Mott critical point, while the photon spectrum remains a sensitive probe~\cite{Inacio2026}. Tensor-network and spectroscopic studies have likewise highlighted the interplay between cavity fields and correlated lattice dynamics~\cite{Nakamoto2025,GrunwaldSpectroscopy2025}. Conversely, realistic multimode calculations indicate that off-resonant changes of correlated exchange are controlled by integrated photonic spectral weight and can be enhanced by confined surface modes~\cite{GrunwaldMultimode2026, Eckhardt2025, bretscher2026}. What is missing is a criterion that connects the frequency and spatial structure of an arbitrary electromagnetic environment to a correlation-driven phase boundary.

Here we derive such a criterion for the Brinkman--Rice~\cite{BrinkmanRice1970} transition and evaluate it for a surface-phonon-polariton geometry. The Gutzwiller wave function and its multiorbital and embedding extensions provide a compact description of local correlations~\cite{Gutzwiller1963,Gutzwiller1964,Gutzwiller1965,Lechermann2007,Lanata2015, LanataYao2014, Chern2017}. Its expectation values become exact within the variational manifold at infinite coordination~\cite{MetznerVollhardt1989,Gebhard1990,Bunemann1998}, yielding the analytic Brinkman--Rice quasiparticle-collapse transition~\cite{BrinkmanRice1970}; intersite interactions can also be incorporated systematically~\cite{Goldstein2020}. Combining this limit with a variational photon displacement, closely related to Lang--Firsov and electron--boson Gutzwiller constructions~\cite{LangFirsov1963,Barone2007,Barone2008,Cui2024, li2023, mazin2024}, yields both a nonperturbative solution for an extensive mode ensemble and a dilute-spectral-weight functional applicable to lossy macroscopic quantum electrodynamics (QED) environments.

We consider the half-filled Hubbard model coupled to normalized electromagnetic modes within the Pauli--Fierz framework,
\begin{align}
 \hat H= &-\frac{t_*}{\sqrt z}\sum_{\langle ij\rangle\sigma}
(c_{i\sigma}^\dagger c_{j\sigma}+\mathrm{H.c.})
+U\sum_i n_{i\uparrow}n_{i\downarrow}\nonumber\\
&+\frac12\sum_\alpha\left[p_\alpha^2+\omega_\alpha^2
\left(q_\alpha-\frac{\lambda_\alpha}{\omega_\alpha}X_\alpha\right)^2\right],
\label{eq:model}\\[-2pt]
X_\alpha= &\sum_i u_{\alpha i}\delta n_i .\nonumber
\end{align}
Here $z$ is the coordination number, and the scaling $t_*/\sqrt z$ keeps the kinetic energy per site finite as $z\to\infty$. The normalized mode profile satisfies $\sum_i|u_{\alpha i}|^2=1$, $\lambda_\alpha$ defines the light--matter coupling, and $\delta n_i$ is the local charge fluctuation. Writing the photonic sector as a completed square retains the gauge-required self-polarization contribution~\cite{Ruggenthaler2014,Flick2017,Schafer2018,Schafer2020}. We characterize each mode by its thermodynamic weight per site and its bond-gradient content,
\begin{equation}
\eta_\alpha=\frac1N\sum_i|u_{\alpha i}|^2,
\quad
r_\alpha=\frac{1}{2\eta_\alpha N_b}
\sum_{\langle ij\rangle}|u_{\alpha i}-u_{\alpha j}|^2,
\label{eq:etar}
\end{equation}
where $N_b=zN/2$ is the number of nearest-neighbor bonds. For the normalization used above, every discrete mode has $\eta_\alpha=1/N$;
we retain $\eta_\alpha$ explicitly because it makes the thermodynamic mode weight and its continuum generalization transparent. The dimensionless factor $r_\alpha$ measures the spatial variation of the mode across nearest-neighbor bonds.
A uniform density mode has $r_\alpha=0$ and is removable at fixed total charge, whereas localized or finite-wave-vector modes generally have $r_\alpha=O(1)$. As shown below, the shift of the Mott boundary depends on this bond-gradient factor. Thus, to contribute to the leading thermodynamic shift, an electromagnetic environment must supply finite local spectral weight and a field profile that varies across the electronic bond.

To derive the analytical shift of the Mott boundary, we use a cavity extension of the Gutzwiller approximation, summarized here and developed in the Supplemental Material (SM). The translationally invariant single-band treatment used in this work is the minimal reduction of a more general cavity extension of the Gutzwiller formalism for intersite correlations. A variational photon displacement maps the Pauli--Fierz Hamiltonian onto a photon-dressed electronic Hamiltonian plus a separable low-rank residual interaction, $\frac12\sum_{\alpha RR'}\hat B_{\alpha R}\hat B_{\alpha R'}$. Its local and connected intersite contributions are evaluated through the extended-Gutzwiller $\mathcal R$ and $\mathcal T$ operator mappings.
In the infinite-coordination limit, this construction evaluates expectation values exactly within the stated variational manifold.

\emph{Mode-extensivity and spectral criterion.--}
If the susceptibility of the normalized cavity-coupled coordinate is non-superextensive, the mode frequencies and couplings remain finite, and no coupled collective coordinate acquires a macroscopic expectation value, each normalized mode changes the total ground-state energy by at most $O(1)$. Consequently, $M$ modes shift the energy density, and hence any regular phase boundary, by $O(M/N)$ (see Theorem~1 in the SM). A fixed or subextensive set of bright modes therefore cannot move a normal-state thermodynamic transition, whereas an extensive set of modes or a continuum with finite local spectral weight can. Exceptions require superradiant order, a divergent susceptibility strong enough to compensate the mode normalization, or superextensive coupling; related gauge constraints have been analyzed in electronic cavity models~\cite{Andolina2019,Guerci2020}.

The recent QMC study of the honeycomb Hubbard model provides a complementary realization of the single-mode limit of this criterion~\cite{Inacio2026}.
There, one properly normalized long-wavelength cavity mode leaves the Gross--Neveu Mott critical point unchanged, consistent with its $O(N^{-1})$ contribution to the energy density.
Moreover, the uniform spin-singlet current channel addressed by the cavity does not overlap the spin-triplet fluctuations controlling the antiferromagnetic transition, so the divergent critical response does not compensate the mode normalization.
Nevertheless, the photon spectrum changes by an order-one amount through the extensive optical conductivity, directly illustrating the distinction between spectroscopic hybridization and thermodynamic control. The framework below extends this single-mode result to multimode and spatially structured environments that can retain finite local spectral weight.

The criterion becomes quantitative through the joint Pauli--Fierz spectral density
\begin{equation}
\Jpf(\omega,r)=\sum_\alpha \eta_\alpha\lambda_\alpha^2
\delta(\omega-\omega_\alpha)\delta(r-r_\alpha).
\label{eq:jpf}
\end{equation}
At dilute cavity-induced spectral weight, independent optimization of the variational displacement for each mode gives
\begin{equation}
\Delta\Uc=-\int_0^\infty d\omega\int dr
\Jpf(\omega,r)\frac{4A_0r}{\omega+4A_0r}
+O(\Jpf^2),
\label{eq:spectral}
\end{equation}
where $A_0=|\ekin|$ and $\ekin$ is the uncorrelated kinetic energy per site. Equation~\ref{eq:spectral} applies two independent filters. The frequency kernel suppresses modes that are too fast to leave a substantial residual self-polarization penalty, while $r$ removes fields that are uniform across the electronic bond. A conventional normalized single mode contributes $\eta_\alpha=1/N$ and therefore vanishes from the energy density in the thermodynamic limit; an extensive continuum can instead retain a finite $\Jpf$.
Equation~\ref{eq:spectral} is the leading term in the cavity-induced spectral weight, whereas the degenerate ensemble considered next can be solved nonperturbatively within the same Gutzwiller manifold.

\begin{figure}[t]
\includegraphics[width=\columnwidth]{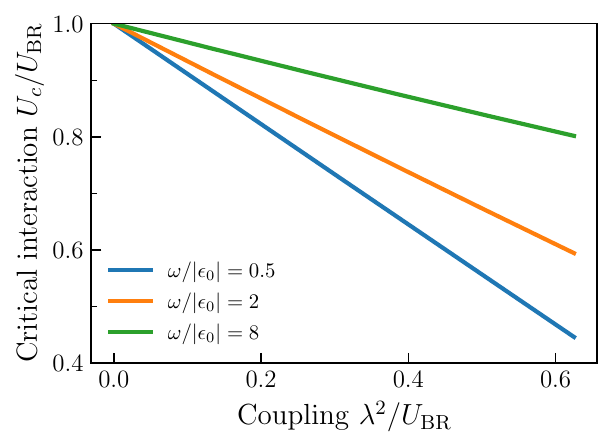}
\vspace{-20pt}
\caption{Cavity-induced displacement of the Brinkman--Rice boundary. Solutions of Eq.~\ref{eq:critical} for $\nmode=r=1$ show that the same gauge-complete coupling lowers $U_c$ through a residual self-polarization penalty for slow modes and bandwidth narrowing for fast modes.
Here $U_{\rm BR}=8A_0$ is the Brinkman--Rice critical interaction in the absence of light--matter coupling.
}
\label{fig:phase}
\end{figure}

For comparison, a degenerate ensemble with mode density $\nmode=M/N$, common frequency $\omega$, common coupling $\lambda$, and common bond factor $r$ is solvable beyond dilute coupling. With a variational displacement $\xi$ and Gaussian squeezing $\sigma$, the infinite-$z$ variational energy density is
\begin{align}
E(D,\xi,\sigma)=&-AZ(D)+[U+\nmode\lambda^2(1-\xi)^2]D\nonumber\\
&+\nmode\frac{\omega}{4}(\sigma+\sigma^{-1}-2),
\label{eq:energy}
\end{align}
with $Z(D)=8D(1-2D)$ and $A=A_0e^{-\nmode r\lambda^2\sigma\xi^2/(2\omega)}\equiv A_0e^{-S}$.
Eq.~\ref{eq:energy} contains three physically distinct contributions: the photon-dressed kinetic energy $-AZ(D)$, the residual self-polarization cost $\nmode\lambda^2(1-\xi)^2D$, and the Gaussian squeezing energy. The first two terms cooperate in favoring localization, but their relative importance changes with frequency.
At the Mott boundary, $D=Z=0$ and $\sigma_c=1$, leading to $A_c=A_0e^{-S_c}$, $\xi_c=\omega/(\omega+4rA_c)$, $S_c=\nmode r\lambda^2\xi_c^2/(2\omega)$, and
\begin{align}
U_c= &8A_c-\nmode\lambda^2(1-\xi_c)^2.
\label{eq:critical}
\end{align}
Expanding Eq.~\ref{eq:critical} to first order in $\nmode$ reproduces Eq.~\ref{eq:spectral} (see the SM for details). The full solution continuously connects the residual self-polarization mechanism for slow modes to photon-induced bandwidth narrowing for fast modes, as shown in Fig.~\ref{fig:phase}.

\begin{figure*}[t]
\includegraphics[width=0.995\textwidth]{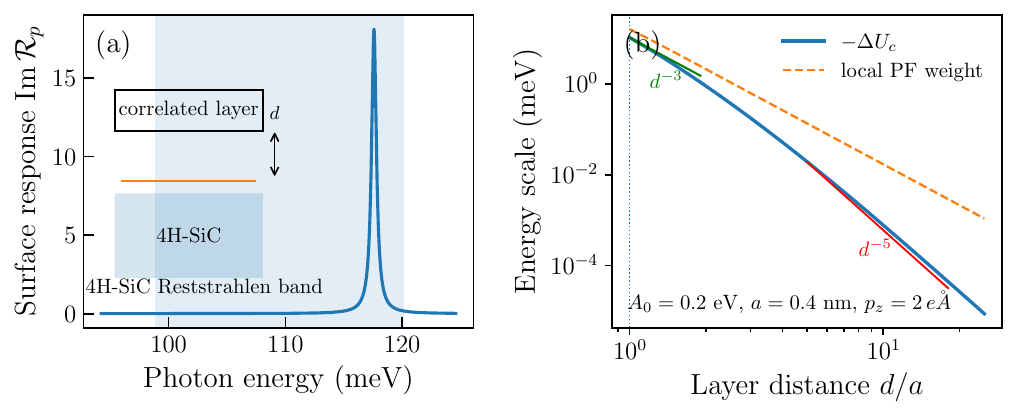}
\vspace{-15pt}
\caption{Bond-scale field variation controls the surface-polariton shift. (a) Geometry and surface response of 4H-SiC. (b) Cavity-induced shift of the Mott boundary and total local Pauli--Fierz weight versus layer--surface distance. Although the local near-field weight grows as $d$ decreases, only its bond-projected component shifts $U_c$, producing a crossover from $d^{-3}$ to $d^{-5}$. Parameters are given in the panel.}
\label{fig:sphp}
\end{figure*}

The interplay with the intrinsic Hubbard interaction is transparent upon defining $u=(U+\delta U_{\rm sp})/(8A)$, with $\delta U_{\rm sp}=\nmode\lambda^2(1-\xi)^2$.
Without the electromagnetic environment, $u_0=U/(8A_0)$, $D=(1-u_0)/4$, $Z=1-u_0^2$, and $U_c^{(0)}=8A_0$. The environment increases the effective correlation ratio in two cooperative ways: residual self-polarization raises the cost of charge fluctuations, while photon dressing lowers the kinetic scale $A$. At fixed microscopic $U$, both effects suppress $D$ and $Z$ and drive the system toward localization. Thus $\Delta U_c<0$ does not mean that the bare Coulomb integral is reduced; it means that a smaller intrinsic $U$ is sufficient to localize the photon-dressed quasiparticles.

\emph{Surface-phonon-polariton environment.--}
Consider a correlated layer a distance $d$ above a planar phonon-polariton surface, such as 4H-SiC, with local polarization $p_z\delta n_i$ normal to the interface. Macroscopic QED expresses the scattering-field spectral density through the dyadic Green tensor~\cite{Dung1998,ScheelBuhmann2008,Philbin2010,Svendsen2021}. In the nonretarded regime, the local Pauli--Fierz weight resolved by in-plane momentum $q$ and photon energy $\Omega$ is
\begin{equation}
\Jpf^{\rm surf}(\Omega,q)=
\frac{p_z^2}{2\pi^2\epsilon_0\Omega}
q^2e^{-2qd}\operatorname{Im}\mathcal R_p(q,\Omega),
\label{eq:surfaceJ}
\end{equation}
where only the cavity-induced scattering Green tensor is retained. For an isotropic bond of length $a$, translational invariance and angular averaging give the bond form factor $r_{\rm b}(q)=1-J_0(qa)$. Substituting Eq.~\ref{eq:surfaceJ} into Eq.~\ref{eq:spectral} gives
\begin{equation}
\Delta U_c^{\rm surf}=-\int_0^\infty d\Omega
\int_0^{q_{\rm max}}dq\,
\Jpf^{\rm surf}(\Omega,q)
\frac{4A_0r_{\rm b}(q)}{\Omega+4A_0r_{\rm b}(q)},
\label{eq:surface_shift}
\end{equation}
which directly incorporates the lossy continuum without discretizing it into an arbitrary set of modes. The cutoff $q_{\rm max}$ is supplied physically by the lattice, Wannier form factors, nonlocal dielectric response, or microscopic surface structure.

Surface phonon polaritons in polar dielectrics provide low-loss, deeply confined infrared fields~\cite{Greffet2002,Hillenbrand2002,Taubner2006,Caldwell2015}. We use the measured single-oscillator dielectric response of 4H-SiC, with $\epsilon_\infty=6.7$, $\omega_{\rm TO}=797~\mathrm{cm}^{-1}$, $\omega_{\rm LO}=969~\mathrm{cm}^{-1}$, and damping $4~\mathrm{cm}^{-1}$~\cite{Tiwald1999}; in the electrostatic limit, $\mathcal R_p=(\epsilon-1)/(\epsilon+1)$. Figure~\ref{fig:sphp} shows that the local spectral weight is concentrated near the surface-phonon-polariton resonance, but the Mott-boundary shift is controlled by its bond-projected component. For $d\gg a$, the dominant momenta satisfy $q\sim d^{-1}$ and $r_{\rm b}(q)\simeq(qa)^2/4$. When $4A_0r_{\rm b}(q)\ll\Omega$ over the relevant spectral window, Eq.~\ref{eq:surface_shift} gives
\begin{equation}
|\Delta U_c|\simeq
\frac{3p_z^2A_0a^2}{8\pi^2\epsilon_0d^5}
\int_0^\infty d\Omega
\frac{\operatorname{Im}\mathcal R_p(\Omega)}{\Omega^2}.
\label{eq:d5}
\end{equation}
For $d\sim a$, before lattice and nonlocal-response cutoffs dominate, $r_{\rm b}(q)$ becomes order unity and the scaling crosses toward the local $d^{-3}$ near-field law. A large local density of optical states is therefore insufficient by itself: the field must vary across a correlated bond. For $A_0=0.2$ eV, $a=0.4$ nm, and a charge-transfer dipole $p_z=2e\mathring{\mathrm A}$, the continuum model gives $|\Delta U_c|=0.37$ meV at $d=1$ nm and $5.0$ meV at $d=0.5$ nm using a lattice cutoff $q\leq\pi/a$. The subnanometer value should be regarded as an upper-scale estimate because nonlocal response, Wannier form factors, and microscopic surface structure then become important.

\emph{Finite-coordination variational Monte Carlo (VMC) validation.--}
We directly test the coefficient that controls the phase-boundary shift. Following variational evaluations of Gutzwiller states~\cite{HongHirsch1990,Foulkes2001}, we sample the determinant state $g^{\hat D}|\Phi_0\rangle$ on finite-$z$ random regular graphs exactly within the chosen variational state up to Monte Carlo error, while integrating the photon coordinates analytically (see the SM for details). For site-local orthonormal modes at weak coupling,
\begin{align}
\frac{\delta E(D)}{\nmode\lambda^2}= &a(D)\xi^2+b(D)(1-\xi)^2,\nonumber\\
\min_\xi\frac{\delta E(D)}{\nmode\lambda^2}= &\frac{a(D)b(D)}{a(D)+b(D)},
\label{eq:vmcfunctional}
\end{align}
where $a(D)$ is obtained from bond-resolved hopping estimators and $b(D)$ from local charge fluctuations. The critical-shift coefficient is the low-$D$ derivative
\begin{equation}
C_{\rm VMC}=\left.\partial_D\min_\xi
\frac{\delta E(D)}{\nmode\lambda^2}\right|_{D\to0},
\quad
\Delta U_c=-\nmode\lambda^2 C_{\rm VMC}.
\label{eq:cvmc}
\end{equation}
This construction does not impose the infinite-$z$ form $Z(D)=8D(1-2D)$ and therefore tests the infinite-coordination contraction. It remains a variational calculation within the Gutzwiller wave-function family rather than an unbiased solution of the finite-dimensional Hubbard model.

\begin{figure*}[t]
\includegraphics[width=0.99\textwidth]{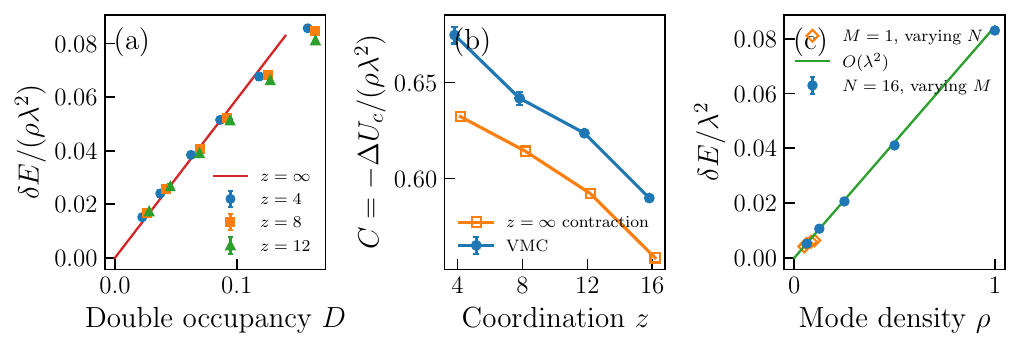}
\vspace{-15pt}
\caption{Finite-coordination validation of the spectral-density criterion. (a) Weak-coupling cavity energy density versus measured double occupancy; its $D\to0$ slope determines the shift of the Mott boundary. (b) The VMC coefficient on $N=24$ random regular graphs agrees with the infinite-coordination prediction within $5$--$7\%$; error bars include graph-to-graph fluctuations. (c) Separate variation of $M$ and $N$ confirms $\delta e\propto M/N$ and the vanishing $1/N$ contribution of a single normalized mode.}
\label{fig:vmcnew}
\end{figure*}

Figure~\ref{fig:vmcnew}(a) shows the sampled cavity energy near the projected limit for $N=24$ and multiple graph realizations. Panel (b) compares the extracted derivative with the analytical coefficient $4A_0/(\omega+4A_0)$ for $z=4,8,12,16$; the difference is only $5$--$7\%$ despite the finite coordination. Panel (c) separately varies $M$ at fixed $N$ and $N$ at fixed $M=1$; both data sets follow the same linear $M/N$ law. An eight-site exact-enumeration self-test, graph-resolved errors, raw samples, and numerical details are provided in the SM.

\emph{Physical interpretation and implications.--}
Equation~\ref{eq:spectral} is governed by the bounded response kernel
\begin{equation}
\mathcal K(\omega,r)=\frac{4A_0r}{\omega+4A_0r}.
\label{eq:kernel}
\end{equation}
For $\omega\ll4A_0r$, a slow electromagnetic coordinate follows the charge fluctuation and leaves the maximal residual self-polarization contribution, $\mathcal K\to1$. For $\omega\gg4A_0r$, the mode instead acts primarily through bandwidth narrowing and $\mathcal K\simeq4A_0r/\omega$. The spatial factor is equally important: a locally intense field that is nearly constant across a bond has $r\to0$ and does not shift the transition. The relevant quantity is therefore neither a resonance frequency nor the local density of optical states alone, but the fluctuation spectrum projected onto the field difference associated with an electronic hop.

Although the realistic environment is supported by a polar lattice, this mechanism is distinct from the conventional attractive interaction of a Hubbard--Holstein model. In the gauge-complete Pauli--Fierz square, self-polarization cancels the zero-frequency boson-exchange interaction, so no static attraction remains at zero frequency. The equilibrium effect instead arises from a dynamical and spatially structured dressing of charge fluctuations and hopping. This distinction also separates collective strong coupling from ground-state control: a long-wavelength bright mode may display an order-one polariton splitting while contributing only $O(1/N)$ to the energy density, whereas a finite-$q$ continuum can produce a finite thermodynamic effect without a single dominant splitting.

The joint spectral density provides a direct cavity-design rule. In a planar structure, $r_{\rm b}(q)=1-J_0(qa)$ converts Eq.~\ref{eq:spectral} into a momentum--frequency integral. In a general nanophotonic structure, the same quantity is obtained from the scattering Green tensor as the local field correlator minus the nearest-neighbor cross correlator, with the free-space contribution subtracted. In the antiadiabatic limit ($\omega\gg4A_0r$),
\begin{equation}
\Delta U_c\simeq-4A_0\int_0^\infty\frac{d\omega}{\omega}
\int dr r\Jpf(\omega,r),
\label{eq:sumrule}
\end{equation}
so the first inverse-frequency moment of the bond-projected spectrum is the relevant figure of merit. Classical electrodynamic solvers can thus screen candidate structures without introducing an arbitrary discrete-mode cutoff.

The SiC geometry gives a particularly direct experimental signature. The evanescent factor $e^{-2qd}$ selects $q\sim1/d$: for $d\sim a$, the field resolves neighboring sites and the shift crosses toward the local $d^{-3}$ behavior, whereas for $d\gg a$ the bond form factor adds $q^2a^2$ and changes the asymptotic scaling to $d^{-5}$. Measuring this crossover by varying spacer thickness would test both the boundary shift and its bond-scale origin. Near a pressure-, gate-, or bandwidth-tuned endpoint, even a sub-meV change in $U_c$ can translate into a measurable displacement of the control parameter. Large charge-transfer dipoles, narrow electronic bands, and broad spectral weight near $q\sim a^{-1}$ and $\omega\lesssim4A_0$ are favorable; thin polar gaps, paired interfaces, and patterned or hyperbolic structures may therefore outperform a single far-field resonance.

The VMC calculation verifies the critical coefficient without imposing the infinite-$z$ form of $Z(D)$ and shows that the $M/N$ extensivity law is robust at finite coordination within the Gutzwiller variational family. It is not, however, an unbiased finite-dimensional Hubbard solution: it omits Hubbard bands and long-range magnetic order. More accurate Dynamical Mean Field Theory (DMFT)~\cite{dmft:2006rmp} or QMC solvers will modify the numerical response kernel while retaining the same electrodynamic input~\cite{Georges1996,Kotliar2006}. Gauge consistency and ultraviolet regularization remain essential: self-polarization must be retained, only the cavity-induced scattering Green tensor should be used to define a change relative to free space, and Wannier form factors, nonlocal dielectric response, and microscopic surface structure provide the physical high-momentum cutoff~\cite{Ruggenthaler2014,Schafer2018,Schafer2020,Svendsen2021}.

In summary, we identify the microscopic conditions under which vacuum electromagnetic fluctuations can alter a bulk correlation-driven transition. The central result is that collective strong coupling alone is not sufficient: except in singular cases involving a divergent susceptibility, macroscopic photon occupation, or superextensive coupling, a normal-state Mott boundary responds only to finite thermodynamic Pauli--Fierz spectral weight that varies across the electronic process controlling localization.
By combining mode extensivity, bond-scale spatial projection, and frequency response in a single joint spectral density, the present framework resolves the apparent tension between large spectroscopic splittings and negligible thermodynamic shifts, while explaining why spatially structured multimode continua can remain effective in the thermodynamic limit. The surface-polariton example converts this principle into an experimentally testable $d^{-3}$-to-$d^{-5}$ crossover, and the finite-coordination VMC calculation shows that the critical coefficient remains quantitatively close to the infinite-coordination result.

More broadly, the joint spectral density establishes a modular bridge between correlated-electron theory and realistic nanophotonic electrodynamics. It allows Green-tensor calculations, classical electromagnetic design, and advanced many-body solvers to be combined without an arbitrary discrete-mode cutoff, providing a practical route to screen and ultimately inverse-design photonic environments for equilibrium phase control. Replacing the Gutzwiller response kernel by DMFT, QMC, or multiorbital response functions should extend the same strategy to charge-transfer, orbital-selective, magnetic, excitonic, and superconducting instabilities. Experimentally, the most promising platforms are narrow-band materials close to a pressure-, gate-, or strain-tuned endpoint, coupled to structures with substantial low-frequency spectral weight at momenta capable of resolving the relevant bond or orbital texture. This work therefore shifts the design objective from maximizing a Rabi splitting or minimizing a mode volume to engineering momentum-, frequency-, and polarization-resolved field correlations matched to the many-body process one aims to control.

\begin{acknowledgments}
We acknowledge support from the US DOE, Office of Science, Basic Energy Sciences, Chemical Sciences, Geosciences, and Biosciences Division under Triad National Security, LLC (``Triad'') contract Grant 89233218CNA000001 (FWP: LANLECF7). 
This research used computational resources provided by the Institutional Computing (IC) Program and the Darwin testbed at Los Alamos National Laboratory (LANL), funded by the Computational Systems and Software Environments subprogram of LANL's Advanced Simulation and Computing program. 
LANL is operated by Triad National Security, LLC, for the National Nuclear Security Administration of the US Department of Energy (Contract No. 89233218CNA000001).
\end{acknowledgments}

\bibliography{references}

\clearpage
\appendix

\setcounter{equation}{0}
\renewcommand{\theequation}{A\arabic{equation}}

\begin{widetext}
\begin{center}
\large\bf Supplemental Material for ``\mytitle"
\end{center}

\input{sm_kernel}

\end{widetext}

\end{document}

%% file: sm_kernel.tex
This Supplemental Material (SM) develops the general cavity-extended Gutzwiller construction, derives the solvable Brinkman--Rice limit and its dilute spectral-density functional, maps the result to a planar surface-phonon-polariton continuum, and documents the finite-coordination determinant-VMC tests.

\section{Cavity-extended Gutzwiller approximation for intersite correlations}
\label{smsec:general_cega}

We first formulate a cavity extension of the Gutzwiller approximation for general intersite correlations and then derive the single-band Brinkman--Rice reduction used in the main text.
An electronic-configuration-dependent photon displacement converts each Pauli--Fierz mode into a photon-dressed electronic Hamiltonian plus a separable low-rank intersite interaction. Such terms lie outside the purely local Gutzwiller treatment but are handled by the extended operator construction for intersite interactions~\cite{Goldstein2020}. The resulting variational structure is closely related to Gutzwiller electron--boson wave functions and variational Lang--Firsov transformations~\cite{LangFirsov1963,Barone2007,Barone2008,Cui2024, li2023, mazin2024}.

\subsection{General electronic and Pauli--Fierz Hamiltonian}

Partition the electronic system into localized correlated units $R$, which may denote sites, correlated orbitals, molecules, or clusters.  We write a Hamiltonian including local correlation and two-site nonlocal correlation~\cite{Goldstein2020},
\begin{equation}
\hat H_{\rm el} =\sum_R\hat H_R^{\rm loc} +\sum_{R\ne R'}\sum_{\mu\nu}
J_{RR'}^{\mu\nu}\hat O_{R\mu}\hat O_{R'\nu},
\label{smeq:Hel_general}
\end{equation}
where $\{\hat O_{R\mu}\}$ is a local operator basis containing densities, orbital transitions, multipoles, spins, and, when needed, pair-changing operators. In the length gauge, the multimode Pauli--Fierz Hamiltonian is
\begin{equation}
\hat H_{\rm PF}=\hat H_{\rm el} +\frac{1}{2}\sum_\alpha
\left[\hat p_\alpha^2+\omega_\alpha^2
\left(\hat q_\alpha-\frac{\hat X_\alpha}{\omega_\alpha}\right)^2\right],
\label{smeq:PF_general}
\end{equation}
with
\begin{equation}
[\hat q_\alpha,\hat p_\beta]=i\delta_{\alpha\beta},
\quad
\hat X_\alpha=\sum_R\hat X_{\alpha R}.
\label{smeq:X_general}
\end{equation}
The completed square retains the dipole self-polarization term required by a gauge-consistent truncated matter description~\cite{Ruggenthaler2014,Flick2017,Schafer2018,Schafer2020}.  Its electronic part is explicitly nonlocal,
\begin{equation}
  \frac{1}{2}\hat X_\alpha^2 =\frac{1}{2}\sum_{RR'}\hat X_{\alpha R}\hat X_{\alpha R'},
\label{smeq:DSE_general}
\end{equation}
but its coupling matrix is separable, with rank no larger than the number of retained electromagnetic modes.

\subsection{Variational displacement and the photon-dressed electronic Hamiltonian}

We use the polaron-dressed Gutzwiller state
\begin{equation}
\ket{\Psi_{\rm cEGA}} =
\exp\!\left[-i\sum_\alpha\hat p_\alpha\hat F_\alpha\right]
\left(\prod_R\hat P_R\right)\ket{\Psi_0} \ket{\chi_{\rm ph}},
\label{smeq:cEGA_state}
\end{equation}
where $\ket{\Psi_0}$ is a Slater determinant or Bogoliubov vacuum, $\hat P_R$ is a local Gutzwiller correlator, and
\begin{equation}
\hat F_\alpha=\sum_R\hat F_{\alpha R},
\quad
\hat F_{\alpha R}=\sum_\mu f_{\alpha R}^{\mu}\hat O_{R\mu}
\label{smeq:F_general}
\end{equation}
is a Hermitian local electronic operator.  The photon state is a centered Gaussian with momentum covariance
\begin{equation}
C^{pp}_{\alpha\beta} =\frac{1}{2} \langle\hat p_\alpha\hat p_\beta+
\hat p_\beta\hat p_\alpha\rangle_{\chi_{\rm ph}}.
\label{smeq:Cpp}
\end{equation}
The variational displacement correlates each electronic configuration with a distinct photonic displacement and therefore captures electron--photon entanglement that is absent from a direct product $\ket{\Psi_G}\ket{\chi_{\rm ph}}$.

The transformation gives
\begin{equation}
\hat q_\alpha\rightarrow\hat q_\alpha+\hat F_\alpha,
\quad
\hat H_{\rm el}\rightarrow \exp\!\left[i\sum_\alpha\hat p_\alpha \mathrm{ad}_{F_\alpha}\right] \hat H_{\rm el},
\label{smeq:transform_general}
\end{equation}
where $\mathrm{ad}_{F}(A)=[F,A]$.  After averaging over the Gaussian photon state,
\begin{equation}
\hat H_{\rm el}^{\rm var} =\overline H_{\rm el}(\bm F,C^{pp})
+\frac{1}{2}\sum_\alpha (\omega_\alpha\hat F_\alpha-\hat X_\alpha)^2,
\label{smeq:Hel_var_general}
\end{equation}
plus the Gaussian photon energy.  If the $\hat F_\alpha$ commute,
\begin{equation}
\overline H_{\rm el} =\exp\!\left[-\frac{1}{2}\sum_{\alpha\beta} C^{pp}_{\alpha\beta}
\mathrm{ad}_{F_\alpha}\mathrm{ad}_{F_\beta}\right]\hat H_{\rm el}.
\label{smeq:Hbar_general}
\end{equation}
Define residual local cavity operators
\begin{equation}
\hat B_{\alpha R}=\omega_\alpha\hat F_{\alpha R}-\hat X_{\alpha R}.
\label{smeq:B_general}
\end{equation}
Then
\begin{equation}
\frac{1}{2}(\omega_\alpha\hat F_\alpha-\hat X_\alpha)^2 =\frac{1}{2}\sum_{RR'}\hat B_{\alpha R}\hat B_{\alpha R'},
\label{smeq:lowrank_general}
\end{equation}
which is a separable low-rank intersite interaction. Because conjugation by a sum of local $\hat F_{\alpha R}$ does not enlarge the support of an initially one- or two-unit electronic term, $\overline H_{\rm el}$ remains compatible with the extended-Gutzwiller representation.

\subsection{Extended-Gutzwiller operator equivalences}

The electronic state in Eq.~\ref{smeq:cEGA_state} is
\begin{equation}
\ket{\Psi_G}=\prod_R\hat P_R\ket{\Psi_0}.
\label{smeq:Gstate_general}
\end{equation}
Introduce quasiparticle fermions $f_{Ra}$ and the one-body density matrix
\begin{equation}
\Delta_{RR',ab}=\bra{\Psi_0} f_{Ra}^\dagger f_{R'b}\ket{\Psi_0}.
\label{smeq:Delta_general}
\end{equation}
The local constraints are
\begin{align}
\bra{\Psi_0}\hat P_R^\dagger\hat P_R\ket{\Psi_0}&=1,\\
\bra{\Psi_0}\hat P_R^\dagger\hat P_R f_{Ra}^\dagger f_{Rb}\ket{\Psi_0} &=\Delta_{R,ab}.
\label{smeq:Gconstraints_general}
\end{align}
Equivalently, one may use a normalized local embedding state $\ket{\Phi_R}$ satisfying
\begin{equation}
\bra{\Phi_R} \hat f_{Rb}\hat f_{Ra}^\dagger\ket{\Phi_R} =\Delta_{R,ab}.
\label{smeq:embedding_constraint_general}
\end{equation}

For a local fermionic creation operator,
\begin{equation}
\hat P_R^\dagger c_{Rm}^\dagger\hat P_R
\longrightarrow\sum_a\mathcal R_{R,ma}f_{Ra}^\dagger,
\label{smeq:Rmap_general}
\end{equation}
where
\begin{equation}
\mathcal R_{R,ma}=\sum_b \bra{\Phi_R}c_{Rm}^\dagger f_{Rb}\ket{\Phi_R}
\left[((\mathbf 1-\Delta_R)\Delta_R)^{-1/2}\right]_{ba}.
\label{smeq:Rmatrix_general}
\end{equation}
For a local Hermitian number-conserving operator $\hat O_{R\mu}$, the leading extended-Gutzwiller equivalence is most transparently written in centered form,
\begin{equation}
\hat P_R^\dagger\hat O_{R\mu}\hat P_R
\longrightarrow
o_{R\mu} +\sum_{ab}\mathcal T_{R\mu;ab} (f_{Ra}^\dagger f_{Rb}-\Delta_{R,ab}),
\label{smeq:Tmap_general}
\end{equation}
with $o_{R\mu}=\bra{\Phi_R}\hat O_{R\mu}\ket{\Phi_R}$.  One convenient explicit representation is
\begin{align}
\mathcal T_{R\mu;dc}=&
\sum_{ab} \left[((\mathbf1-\Delta_R)\Delta_R)^{-1/2}\right]_{da}
\bra{\Phi_R}\hat O_{R\mu}f_{Rb}f_{Ra}^\dagger\ket{\Phi_R}
\left[((\mathbf1-\Delta_R)\Delta_R)^{-1/2}\right]_{bc}
\nonumber\\
&-o_{R\mu}\left[(\mathbf1-\Delta_R)^{-1}\right]_{dc}.
\label{smeq:Tmatrix_general}
\end{align}
For $R\ne R'$, Wick contraction in $\ket{\Psi_0}$ gives
\begin{equation}
\langle\hat O_{R\mu}\hat O_{R'\nu}\rangle_G
\simeq o_{R\mu}o_{R'\nu} -\Tr[\mathcal T_{R\mu}\Delta_{RR'}
\mathcal T_{R'\nu}\Delta_{R'R}].
\label{smeq:two_unit_general}
\end{equation}
Local products with $R=R'$ are evaluated directly in $\ket{\Phi_R}$.  Pair-changing and correlated-hopping channels are represented by the corresponding $\mathcal S$ and $\mathcal U$ tensors of the general intersite construction~\cite{Goldstein2020}; they are not needed for the density-coupled single-band reduction below.

\subsection{Cavity-extended energy functional and saddle-point equations}

Application of Eq.~\ref{smeq:two_unit_general} to the residual operators in Eq.~\ref{smeq:B_general} gives
\begin{align}
E_{\rm res}^{\rm cav}=&\frac{1}{2}\sum_\alpha\Bigg[
\sum_R\bra{\Phi_R}\hat B_{\alpha R}^2\ket{\Phi_R}
+\sum_{R\ne R'}\Big(b_{\alpha R}b_{\alpha R'}
-\Tr[\mathcal T^B_{\alpha R}\Delta_{RR'}
\mathcal T^B_{\alpha R'}\Delta_{R'R}]
\Big)\Bigg],
\label{smeq:Ecav_general}
\end{align}
where $b_{\alpha R}=\bra{\Phi_R}\hat B_{\alpha R}\ket{\Phi_R}$.  The first intersite contribution is a correlated Hartree term, whereas the second describes connected quasiparticle coherence.  The full constrained variational principle is therefore
\begin{equation}
E_0^{\rm cEGA}=\min_{\Psi_0,\{\Phi_R\},\{f_{\alpha R}^{\mu}\},C^{pp}}
\left\{E_{\rm extGA}\!\left[
\overline H_{\rm el}+\frac{1}{2}\sum_{\alpha RR'}
\hat B_{\alpha R}\hat B_{\alpha R'}\right]+E_{\rm ph}[C^{pp}]\right\},
\label{smeq:variational_general}
\end{equation}
subject to Eqs.~\ref{smeq:Gconstraints_general} and \ref{smeq:embedding_constraint_general}.  Variation with respect to $\ket{\Psi_0}$ produces a quasiparticle Hamiltonian
\begin{equation}
[ h_{\rm QP}]_{RR',ab} =\frac{\partial E_{\rm cEGA}}{\partial\Delta_{R'R,ba}},
\label{smeq:QP_general}
\end{equation}
including a low-rank nonlocal cavity contribution.  Variation with respect to $\ket{\Phi_R}$ gives a local embedding problem containing the usual hybridization term together with a cavity-induced impurity--bath vertex of schematic form
\begin{equation}
\sum_{\alpha,ab}\mathcal F_{\alpha R;ab}
\hat B_{\alpha R}(f_{Rb}f_{Ra}^\dagger-\Delta_{R,ab}).
\label{smeq:embedding_vertex_general}
\end{equation}
Finally, $\partial E/\partial f_{\alpha R}^{\mu}=0$ and variation with respect to $C^{pp}$ determine the optimal variational displacement and Gaussian photon covariance. Equations~\ref{smeq:Ecav_general}--\ref{smeq:variational_general} define the cavity extension of the intersite Gutzwiller formalism.

\subsection{Reduction to the density-coupled single-band model}
\label{smsec:reduction_general}

For the Hubbard model, choose
\begin{equation}
\hat X_\alpha=\lambda_\alpha\sum_i u_{\alpha i}\delta n_i,
\quad
\hat F_\alpha=\frac{\lambda_\alpha\xi_\alpha}{\omega_\alpha}
\sum_i u_{\alpha i}\delta n_i.
\label{smeq:density_reduction}
\end{equation}
The transformed hopping is
\begin{equation}
c_{i\sigma}^\dagger c_{j\sigma}\rightarrow
c_{i\sigma}^\dagger c_{j\sigma} \exp\!\left[i\sum_\alpha
\frac{\lambda_\alpha\xi_\alpha}{\omega_\alpha}
\hat p_\alpha(u_{\alpha i}-u_{\alpha j})\right],
\label{smeq:hopping_reduction}
\end{equation}
whereas the residual interaction is
\begin{equation}
\hat H_{\rm res}=\frac{1}{2}\sum_{ij}K_{ij}\delta n_i\delta n_j,
\quad
K_{ij}=\sum_\alpha\lambda_\alpha^2(1-\xi_\alpha)^2
u_{\alpha i}u_{\alpha j}.
\label{smeq:K_general}
\end{equation}
In the translationally averaged, half-filled, paramagnetic infinite-coordination limit, the connected intersite term in Eq.~\ref{smeq:Ecav_general} is subleading for the isotropic mode ensemble used in the main text.  The general cEGA functional then reduces to three scalar variables: the double occupancy $D$, displacement fraction $\xi$, and squeezing $\sigma$.  The following sections derive this reduction and then restore a general electromagnetic spectrum to leading order in its local spectral weight.

\section{Reduction to the solvable Brinkman--Rice problem}
\label{smsec:scope_reduced}

The main text specializes the general functional of Sec.~\ref{smsec:general_cega} to the zero-temperature, half-filled, paramagnetic Gutzwiller state at infinite coordination. In this limit the Gutzwiller approximation evaluates expectation values exactly within the chosen variational manifold~\cite{MetznerVollhardt1989,Bunemann1998}. The resulting Brinkman--Rice transition is therefore exact \emph{within that manifold}, but it is not the exact Hubbard-model transition: dynamical mean-field theory additionally retains incoherent Hubbard bands, dynamical self-energy effects, and superexchange processes absent from the elementary Brinkman--Rice state.

The mode-extensivity theorem below assumes a normal, nonsuperradiant phase in which the connected correlation and response kernels of the normalized cavity-coupled coordinates remain non-superextensive. It does not exclude a finite effect from one mode when that mode overlaps a macroscopic order parameter, when a divergent susceptibility compensates mode normalization, or when the light--matter coupling is scaled superextensively.

\section{Model and mode geometry}

We consider
\begin{equation}
\hat H=\hat H_{\rm H}+\sum_{\alpha=1}^{M}\hat H_{\alpha},
\end{equation}
with
\begin{equation}
\hat H_{\rm H}=-\frac{t_*}{\sqrt z}\sum_{\langle ij\rangle\sigma}
(c_{i\sigma}^\dagger c_{j\sigma}+\mathrm{H.c.})
+U\sum_i n_{i\uparrow}n_{i\downarrow},
\label{smeq:hubbard}
\end{equation}
and
\begin{equation}
\hat H_{\alpha}=\frac{1}{2}\left[\hat p_\alpha^2+\omega_\alpha^2
\left(\hat q_\alpha-\frac{\lambda_\alpha}{\omega_\alpha}\hat X_\alpha\right)^2\right],
\quad
\hat X_\alpha=\sum_i u_{\alpha i}\delta\hat n_i,
\label{smeq:pf}
\end{equation}
where $\delta\hat n_i=\hat n_i-1$ and $\sum_i|u_{\alpha i}|^2=1$. We set $\hbar=1$ except in the macroscopic-QED mapping of Sec.~\ref{smsec:sphp}. Mode profiles are taken real for notational simplicity; the complex case follows by inserting the corresponding conjugates. For each mode define
\begin{equation}
\eta_\alpha=\frac{1}{N}\sum_i|u_{\alpha i}|^2,
\quad
r_\alpha=\frac{1}{2\eta_\alpha N_b}
\sum_{\langle ij\rangle}|u_{\alpha i}-u_{\alpha j}|^2,
\label{smeq:eta_r_modes}
\end{equation}
with $N_b=zN/2$. For a conventionally normalized discrete mode, $\eta_\alpha=1/N$ identically. The thermodynamic mode density and its bond-weighted geometry factor are
\begin{equation}
\nmode=\sum_\alpha\eta_\alpha=\frac{M}{N},
\qquad
r=\frac{1}{\nmode}\sum_\alpha\eta_\alpha r_\alpha
=\frac{1}{2\nmode N_b}
\sum_{\langle ij\rangle,\alpha}|u_{\alpha i}-u_{\alpha j}|^2.
\label{smeq:rho_r_modes}
\end{equation}
Thus $\nmode$ is the thermodynamic weight per site of the normalized mode ensemble, whereas $r$ measures field variation across a nearest-neighbor bond. A uniform density mode has $r_\alpha=0$; finite-wave-vector and localized modes generally retain $r_\alpha=O(1)$.

\section{Normal-state mode-extensivity theorem}

\begin{theorem}[Normal-state mode-extensivity criterion]
Let $\hat H_N^{(0)}$ be an extensive short-range matter Hamiltonian on $N$ correlated units, and let $M$ normalized harmonic modes couple through $\hat X_\alpha=\sum_i u_{\alpha i}\hat O_i$, with bounded local $\hat O_i$ and $\sum_i|u_{\alpha i}|^2=1$. Assume that (i) no $\hat X_\alpha$ acquires a macroscopic expectation value, (ii) the connected equal-time correlation matrix and the corresponding static response kernel have $O(1)$ operator norms in the mode channels, (iii) the mode frequencies remain finite and bounded away from zero, and (iv) the coupling amplitudes do not grow with $N$. Then
\begin{equation}
E_N-E_N^{(0)}=O(M),
\quad
\frac{E_N-E_N^{(0)}}{N}=O\left(\frac{M}{N}\right).
\label{smeq:theorem}
\end{equation}
If a regular phase boundary is determined by equality of two extensive energy densities with a nonzero derivative with respect to its tuning parameter, its displacement is also $O(M/N)$.
\end{theorem}

\begin{proof}
For a normalized mode,
\begin{equation}
\langle\hat X_\alpha^2\rangle_c
=\sum_{ij}u_{\alpha i}u_{\alpha j}
\langle\hat O_i\hat O_j\rangle_c.
\end{equation}
The assumed $O(1)$ norm of the connected correlation kernel implies $\langle\hat X_\alpha^2\rangle_c=O(1)$. The optimized completed Pauli--Fierz square therefore contributes at most $O(1)$ per mode. The photon-induced dressing of any short-range electronic term is governed by normalized field differences and the same bounded response kernel, and is likewise $O(1)$ per mode. Summing over $M$ modes gives Eq.~\ref{smeq:theorem}. If $f(g)$ is the difference between two uncoupled energy densities, with $f(g_c)=0$ and $f'(g_c)\neq0$, an $O(M/N)$ correction shifts the root by the same order through the implicit-function theorem.
\end{proof}

\begin{remark}
At a critical point, a response-kernel eigenvalue may diverge. One normalized mode can then produce a finite effect if its profile overlaps the critical eigenvector strongly enough to compensate normalization. The theorem is therefore a normal-state extensivity statement, not a universal no-go theorem for critical or superradiant settings. The Brinkman--Rice transition studied here is a local quasiparticle-collapse transition within the paramagnetic Gutzwiller manifold and obeys the stated counting.
\end{remark}


\section{Variational-displacement Gutzwiller state}

We use
\begin{equation}
\ket{\Psi(D,\xi,\sigma)} = \hat U_\xi\prod_i\hat P_i\ket{\Phi_0}\ket{\chi_\sigma},
\end{equation}
where
\begin{equation}
\hat U_\xi= \exp\left[-i\frac{\lambda\xi}{\omega}
\sum_{\alpha i}\hat p_\alpha u_{\alpha i}\delta\hat n_i\right].
\label{smeq:U}
\end{equation}
The photon Gaussian obeys
\begin{equation}
\avg{q_\alpha}=\avg{p_\alpha}=0,
\quad
\avg{p_\alpha^2}=\frac{\omega\sigma}{2},
\quad
\avg{q_\alpha^2}=\frac{1}{2\omega\sigma}.
\label{smeq:gaussian}
\end{equation}
The bare vacuum has $\sigma=1$.  The Gutzwiller local probabilities at half filling are
\begin{equation}
p_0=p_2=D,
\quad
p_\uparrow=p_\downarrow=\frac{1}{2}-D.
\label{smeq:probabilities}
\end{equation}
The quasiparticle weight is
\begin{equation}
Z(D)=8D(1-2D).
\label{smeq:Z}
\end{equation}

\subsection{Photon transformation of hopping}

Using $[\delta n_i,c_{i\sigma}^\dagger]=c_{i\sigma}^\dagger$, one obtains
\begin{equation}
\hat U_\xi^\dagger c_{i\sigma}^\dagger c_{j\sigma}\hat U_\xi
=c_{i\sigma}^\dagger c_{j\sigma}
\exp\left[i\frac{\lambda\xi}{\omega}
\sum_\alpha \hat p_\alpha(u_{\alpha i}-u_{\alpha j})\right].
\end{equation}
The Gaussian average is
\begin{align}
\left\langle e^{i(\lambda\xi/\omega)\sum_\alpha \hat p_\alpha
(u_{\alpha i}-u_{\alpha j})}\right\rangle_{\chi_\sigma}
&=\exp\left[-\frac{\lambda^2\xi^2\sigma}{4\omega}
\sum_\alpha|u_{\alpha i}-u_{\alpha j}|^2\right].
\label{smeq:debye}
\end{align}
Averaging over bonds and using Eq.~\ref{smeq:rho_r_modes} gives the dressed kinetic scale
\begin{equation}
A=A_0\exp\left[-\nmode r\frac{\lambda^2\xi^2\sigma}{2\omega}\right],
\quad A_0=|\ekin|,
\label{smeq:A}
\end{equation}
where $\ekin<0$ is the noninteracting kinetic energy per site.

\subsection{Residual self-polarization term}

The displacement shifts
\begin{equation}
\hat U_\xi^\dagger \hat q_\alpha\hat U_\xi
=\hat q_\alpha+\frac{\lambda\xi}{\omega}\sum_i u_{\alpha i}\delta n_i.
\end{equation}
Consequently,
\begin{equation}
\hat U_\xi^\dagger\hat H_\alpha\hat U_\xi
=\frac{1}{2}\left[\hat p_\alpha^2+\omega^2
\left(\hat q_\alpha-\frac{\lambda(1-\xi)}{\omega}X_\alpha\right)^2\right].
\end{equation}
For the isotropic ensemble used in the main text, the diagonal part of the mode kernel is $\nmode$ and off-diagonal density contractions are subleading in the infinite-coordination normal phase.  Therefore
\begin{equation}
\frac{1}{N}\sum_\alpha\frac{\lambda^2(1-\xi)^2}{2} \avg{X_\alpha^2}
=\nmode\lambda^2(1-\xi)^2D.
\label{smeq:residual}
\end{equation}
We used
\begin{equation}
\avg{(\delta n_i)^2}_G=2D.
\end{equation}
For a general finite-range kernel, the extended-Gutzwiller equivalences of Ref.~\cite{Goldstein2020} add explicit connected intersite terms. These terms preserve the subextensive-versus-extensive counting but modify the geometry-dependent coefficient and the quantitative critical line.

\subsection{Squeezing energy}

For one pure Gaussian mode,
\begin{equation}
\frac{1}{2}\avg{p^2}+\frac{\omega^2}{2}\avg{q^2}-\frac{\omega}{2}
=\frac{\omega}{4}(\sigma+\sigma^{-1}-2).
\end{equation}
There are $M=\nmode N$ effective modes, so the squeezing energy density is
\begin{equation}
E_{\rm sq}=\nmode\frac{\omega}{4}(\sigma+\sigma^{-1}-2).
\label{smeq:sq}
\end{equation}

\subsection{Infinite-coordination variational energy functional}

Combining Eqs.~\ref{smeq:Z}, \ref{smeq:A}, \ref{smeq:residual}, and \ref{smeq:sq}, the variational energy per site is
\begin{equation}
E(D,\xi,\sigma)=-A Z(D)+[U+\nmode\lambda^2(1-\xi)^2]D
+\nmode\frac{\omega}{4}(\sigma+\sigma^{-1}-2).
\label{smeq:energy}
\end{equation}
Define
\begin{equation}
U_{\rm eff}=U+\nmode\lambda^2(1-\xi)^2,
\quad
u=\frac{U_{\rm eff}}{8A}.
\label{smeq:u}
\end{equation}
Stationarity with respect to $D$ gives
\begin{equation}
-8A(1-4D)+U_{\rm eff}=0,
\end{equation}
so that
\begin{equation}
D=\frac{1-u}{4},
\quad
Z=1-u^2,
\label{smeq:DZ}
\end{equation}
for $u<1$, while $D=Z=0$ for $u\geq1$.  Substitution yields
\begin{equation}
E_{\rm el}^{\star}=-A(1-u)^2.
\label{smeq:emin}
\end{equation}

\subsection{Stationarity conditions for displacement and squeezing}

Using the envelope theorem, derivatives with respect to $\xi$ and $\sigma$ may be taken at fixed optimized $D$.  Since
\begin{equation}
\frac{\partial A}{\partial\xi}=-2\nmode r\frac{\lambda^2}{2\omega}\sigma\xi A,
\end{equation}
we find
\begin{equation}
\nmode r\frac{\lambda^2}{\omega}\sigma\xi A Z
=2\nmode\lambda^2(1-\xi)D.
\end{equation}
Using $Z/D=4(1+u)$ gives
\begin{equation}
\xi=\frac{\omega}{\omega+2r\sigma A(1+u)}.
\label{smeq:xi}
\end{equation}
Similarly,
\begin{equation}
\frac{\partial A}{\partial\sigma}=-\nmode r\frac{\lambda^2}{2\omega}\xi^2A,
\end{equation}
and $\partial e/\partial\sigma=0$ gives
\begin{equation}
\sigma=\left[1+\frac{2r\lambda^2\xi^2AZ}{\omega^2}\right]^{-1/2}.
\label{smeq:sigma}
\end{equation}
The factors of $\nmode$ cancel from the explicit stationarity equations, but remain in $A$ and $u$.

\section{Critical line}

At the Brinkman--Rice boundary, $u_c=1$, $D_c=Z_c=0$, and the squeezing returns to the vacuum value $\sigma_c=1$.  Equations~\ref{smeq:A} and \ref{smeq:xi} become
\begin{align}
A_c&=A_0\exp\left[-\nmode r\frac{\lambda^2}{2\omega}\xi_c^2\right],
\label{smeq:Ac}\\
\xi_c&=\frac{\omega}{\omega+4rA_c}.
\label{smeq:xic}
\end{align}
The condition $u_c=1$ gives
\begin{equation}
\Uc=8A_c-\nmode\lambda^2(1-\xi_c)^2.
\label{smeq:Uc}
\end{equation}
These three scalar equations give the Brinkman--Rice boundary exactly within the infinite-coordination variational manifold for the solvable mode ensemble.

\subsection{Adiabatic limit}

For $\omega\ll4rA_0$,
\begin{equation}
\xi_c\simeq\frac{\omega}{4rA_0},
\quad
A_c=A_0+O(\omega),
\end{equation}
so that
\begin{equation}
\Uc=8A_0-\nmode\lambda^2+O(\omega).
\label{smeq:adiabatic}
\end{equation}

\subsection{Antiadiabatic limit}

For $\omega\gg4rA_0$,
\begin{equation}
\xi_c=1-\frac{4rA_0}{\omega}+O(\omega^{-2}),
\end{equation}
and the residual term is higher order.  Therefore
\begin{equation}
\Uc=8A_0\exp\left[-\nmode r\frac{\lambda^2}{2\omega}\right]
+O(\omega^{-2}).
\label{smeq:anti}
\end{equation}

\subsection{Dilute-mode expansion of the critical interaction}
\label{smsec:dilute_expansion}

Here we show explicitly that the dilute-mode expansion of the critical equations above produces the bounded kernel used in the main text and in the general spectral functional derived below. We define the cavity-induced shift of the
Brinkman--Rice critical interaction as
\begin{equation}
\Delta U_c \equiv U_c-U_{\rm BR},
\quad
U_{\rm BR}=8A_0.
\label{smeq:delta_uc_def}
\end{equation}
For a degenerate ensemble with mode density $\nmode=M/N$, the critical
equations are
\begin{align}
A_c &= A_0 e^{-S_c},
\label{smeq:ac_sm}\\
S_c &= \frac{\nmode r\lambda^2}{2\omega}\xi_c^2,
\label{smeq:sc_sm}\\
\xi_c &= \frac{\omega}{\omega+4rA_c},
\label{smeq:xic_sm}\\
U_c &= 8A_c-\nmode\lambda^2(1-\xi_c)^2.
\label{smeq:uc_sm}
\end{align}

At $\nmode=0$, the cavity does not modify the electronic system, and
therefore
\begin{equation}
A_c^{(0)}=A_0,
\quad
U_c^{(0)}=8A_0.
\label{smeq:zeroth_sm}
\end{equation}
The zeroth-order variational displacement is consequently
\begin{equation}
\xi_0
=
\frac{\omega}{\omega+4rA_0}.
\label{smeq:xi0_sm}
\end{equation}
Since $A_c-A_0=O(\nmode)$, one has
\begin{equation}
\xi_c=\xi_0+O(\nmode).
\label{smeq:xi_expand_sm}
\end{equation}
Both $S_c$ and the residual self-polarization term in
Eq.~\ref{smeq:uc_sm} already contain an explicit factor of $\nmode$.
Thus, to first order in $\nmode$,
\begin{align}
\nmode\xi_c^2 &= \nmode\xi_0^2+O(\nmode^2),
\label{smeq:rho_xi_sm}\\
\nmode(1-\xi_c)^2 &= \nmode(1-\xi_0)^2+O(\nmode^2).
\label{smeq:rho_residual_sm}
\end{align}
The explicit $O(\nmode)$ correction to $\xi_c$ is therefore not needed
for determining $\Delta U_c$ to linear order.

Because $S_c=O(\nmode)$, the renormalized kinetic scale can be expanded as
\begin{equation}
A_c = A_0e^{-S_c} =A_0(1-S_c)+O(\nmode^2).
\label{smeq:ac_expand_sm}
\end{equation}
Using Eq.~\ref{smeq:xi0_sm}, this gives
\begin{equation}
A_c-A_0 = - A_0 \frac{\nmode r\lambda^2}{2\omega}
\xi_0^2 +O(\nmode^2).
\label{smeq:ac_shift_sm}
\end{equation}
The corresponding contribution to the shift of the critical
interaction is
\begin{align}
\Delta U_c^{\rm BW} &\equiv 8(A_c-A_0)
\nonumber\\
&= - \frac{4\nmode A_0r\lambda^2}{\omega} \xi_0^2 +O(\nmode^2)
\nonumber\\
&= -\nmode\lambda^2 \frac{4A_0r \omega} {(\omega+4A_0r)^2} +O(\nmode^2).
\label{smeq:delta_uc_bw_sm}
\end{align}
This term originates from the photon-induced narrowing of the
electronic bandwidth.

The residual self-polarization contribution is
\begin{align}
\Delta U_c^{\rm SP} &\equiv -\nmode\lambda^2(1-\xi_c)^2
\nonumber\\
&= -\nmode\lambda^2(1-\xi_0)^2 +O(\nmode^2).
\label{smeq:delta_uc_sp_1_sm}
\end{align}
Using
\begin{equation}
1-\xi_0 = \frac{4A_0r}{\omega+4A_0r},
\label{smeq:one_minus_xi_sm}
\end{equation}
one obtains
\begin{equation}
\Delta U_c^{\rm SP} = -\nmode\lambda^2
\frac{(4A_0r)^2}  {(\omega+4A_0r)^2} +O(\nmode^2).
\label{smeq:delta_uc_sp_sm}
\end{equation}

Adding Eqs.~\ref{smeq:delta_uc_bw_sm} and
\ref{smeq:delta_uc_sp_sm}, the total shift becomes
\begin{align}
\Delta U_c &=-\nmode\lambda^2
\frac{4A_0r \omega+(4A_0r)^2}  {(\omega+4A_0r)^2} +O(\nmode^2)
\nonumber\\
&= -\nmode\lambda^2
\frac{4A_0r(\omega+4A_0r)} {(\omega+4A_0r)^2} +O(\nmode^2)
\nonumber\\
&= -\nmode\lambda^2 \frac{4A_0r} {\omega+4A_0r} +O(\nmode^2).
\label{smeq:delta_uc_final_sm}
\end{align}
Thus, although the bandwidth-narrowing and residual self-polarization
contributions separately contain squared denominators, their sum
reduces to the single kernel appearing in Eq.~\ref{smeq:spectralfunctional}.

To make the connection to the joint Pauli--Fierz spectral density
explicit, consider $M$ degenerate extended modes with common
$\omega$, $r$, and $\lambda$. For a conventionally normalized mode,
$\eta_\alpha=1/N$, so that
\begin{equation}
\sum_{\alpha=1}^{M}\eta_\alpha = \frac{M}{N} =\nmode.
\label{smeq:eta_rho_sm}
\end{equation}
The corresponding spectral density is
\begin{equation}
\mathcal{J}_{\rm PF}(\omega',r') =
\nmode\lambda^2 \delta(\omega'-\omega) \delta(r'-r).
\label{smeq:jpf_degenerate_sm}
\end{equation}
Substituting Eq.~\ref{smeq:jpf_degenerate_sm} into the spectral functional yields
\begin{align}
\Delta U_c &= -\int_0^\infty d\omega' \int dr'
\mathcal{J}_{\rm PF}(\omega',r')
\frac{4A_0r'}{\omega'+4A_0r'} +O(\mathcal{J}_{\rm PF}^2)
\nonumber\\
&= -\nmode\lambda^2 \frac{4A_0r}{\omega+4A_0r} +O(\nmode^2),
\label{smeq:spectral_degenerate_sm}
\end{align}
in agreement with Eq.~\ref{smeq:delta_uc_final_sm}.

For a dilute collection of nondegenerate modes, the first-order
contributions are additive:
\begin{equation}
\Delta U_c = -\sum_\alpha
\eta_\alpha\lambda_\alpha^2
\frac{4A_0r_\alpha}  {\omega_\alpha+4A_0r_\alpha}
+O(\mathcal{J}_{\rm PF}^2).
\label{smeq:multimode_sum_sm}
\end{equation}
Using the definition
\begin{equation}
\mathcal{J}_{\rm PF}(\omega,r)
= \sum_\alpha \eta_\alpha\lambda_\alpha^2
\delta(\omega-\omega_\alpha) \delta(r-r_\alpha),
\label{smeq:jpf_repeat_sm}
\end{equation}
Eq.~\ref{smeq:multimode_sum_sm} is immediately converted into Eq.~\ref{smeq:spectralfunctional}. Terms of
$O(\mathcal{J}_{\rm PF}^2)$ arise from the common self-consistent
renormalization of $A_c$: the change of $A_c$ induced by one mode
modifies the optimized variational displacement of the other modes,
generating products of their spectral weights.

The limiting behavior also separates the two physical mechanisms. For
slow modes,
\begin{equation}
\omega\ll 4A_0r:
\quad
\Delta U_c \simeq -\nmode\lambda^2,
\label{smeq:slow_limit_sm}
\end{equation}
and the residual self-polarization penalty dominates. For fast modes,
\begin{equation}
\omega\gg 4A_0r:
\quad
\Delta U_c \simeq -\nmode\lambda^2\frac{4A_0r}{\omega},
\label{smeq:fast_limit_sm}
\end{equation}
where the residual self-polarization contribution is only
$O(\omega^{-2})$, while the leading $O(\omega^{-1})$ correction comes
from photon-induced bandwidth narrowing.

\section{Photon number and critical cusp}

For a mode with annihilation operator
\begin{equation}
a_\alpha=\sqrt{\frac{\omega}{2}}q_\alpha+
\frac{i}{\sqrt{2\omega}}p_\alpha,
\end{equation}
the squeezed-vacuum contribution is
\begin{equation}
\avg{a_\alpha^\dagger a_\alpha}_{\rm sq}
=\frac14(\sigma+\sigma^{-1}-2).
\end{equation}
The variational displacement contributes
\begin{equation}
\frac{\omega}{2}\avg{F_\alpha^2},
\quad
F_\alpha=\frac{\lambda\xi}{\omega}\sum_i u_{\alpha i}\delta n_i.
\end{equation}
For the isotropic ensemble,
\begin{equation}
\frac{N_{\rm ph}}{N}=\nmode\frac{\sigma+\sigma^{-1}-2}{4}
+\nmode\frac{\lambda^2}{\omega}\xi^2D.
\label{smeq:nph}
\end{equation}
Near the transition, $D\sim(\Uc-U)$ and $\sigma-1=O(D)$.  The squeezing occupation is therefore $O(D^2)$, whereas the displacement occupation is $O(D)$.  Thus
\begin{equation}
\frac{N_{\rm ph}}{N}=C_{\rm ph}(\Uc-U)+O[(\Uc-U)^2]
\end{equation}
inside the metal and vanishes in the Brinkman--Rice insulator.  The photon density is therefore continuous but has a cusp at the transition. Standard input--output relations can convert the associated intracavity correlations into experimentally accessible output spectra~\cite{GrunwaldSpectroscopy2025}.

\section{General nondegenerate spectral-density functional}

We now allow every mode to have its own frequency, coupling, spatial profile, and variational displacement. For the normalization used here $\eta_\alpha=1/N$, but we retain it to make the thermodynamic weighting and continuum limit explicit. Define
\begin{equation}
\eta_\alpha=\frac1N\sum_i|u_{\alpha i}|^2,
\quad
r_\alpha=\frac{1}{2\eta_\alpha N_b}
\sum_{\langle ij\rangle}|u_{\alpha i}-u_{\alpha j}|^2.
\label{smeq:etaralpha}
\end{equation}
To first order in the total local photonic spectral weight, the variational optimization separates mode by mode.  The transition-point bandwidth is
\begin{equation}
A_c=A_0\exp\left[-\frac{1}{2}\sum_\alpha
\eta_\alpha r_\alpha\frac{\lambda_\alpha^2}{\omega_\alpha}
\xi_\alpha^2\right]+O(\lambda^4),
\label{smeq:Anondeg}
\end{equation}
and the critical interaction is
\begin{equation}
U_c=8A_c-\sum_\alpha\eta_\alpha\lambda_\alpha^2
(1-\xi_\alpha)^2+O(\lambda^4).
\label{smeq:Unondeg}
\end{equation}
The stationarity condition for mode $\alpha$, evaluated at the uncoupled transition, is
\begin{equation}
\xi_{\alpha 0}=\frac{\omega_\alpha}
{\omega_\alpha+4A_0r_\alpha}.
\label{smeq:xialpha}
\end{equation}
Expanding Eqs.~\ref{smeq:Anondeg} and \ref{smeq:Unondeg} and using Eq.~\ref{smeq:xialpha} gives
\begin{equation}
\Delta U_c=-\sum_\alpha\eta_\alpha\lambda_\alpha^2
\frac{4A_0r_\alpha}{\omega_\alpha+4A_0r_\alpha}
+O(\lambda^4).
\label{smeq:generalshift}
\end{equation}
The two terms in Eq.~\ref{smeq:Unondeg} combine into a single bounded response factor: the bandwidth term contributes
$-(4A_0r_\alpha/\omega_\alpha)\xi_{\alpha0}^2$, while the residual square contributes $-(1-\xi_{\alpha0})^2$; their sum is $-(1-\xi_{\alpha0})$.

The natural joint spectral density is
\begin{equation}
\mathcal J_{\rm PF}(\omega,r)=
\sum_\alpha\eta_\alpha\lambda_\alpha^2
\delta(\omega-\omega_\alpha)\delta(r-r_\alpha).
\label{smeq:JPF}
\end{equation}
Equation~\ref{smeq:generalshift} then becomes
\begin{equation}
\Delta U_c=-\int_0^\infty d\omega\int dr
\mathcal J_{\rm PF}(\omega,r)
\frac{4A_0r}{\omega+4A_0r}
+O(\mathcal J_{\rm PF}^2).
\label{smeq:spectralfunctional}
\end{equation}
A frequency-only spectral density is insufficient because two modes at the same frequency may have parametrically different bond gradients. Equation~\ref{smeq:spectralfunctional} is therefore a joint frequency--geometry functional. A conventional normalized bright mode has $\eta_\alpha=1/N$ and produces an $O(N^{-1})$ shift, whereas a lossy continuum can retain finite $\mathcal J_{\rm PF}$ without being discretized into normal modes.

If a mode is written in the usual form
$\hbar g_\alpha(a_\alpha+a_\alpha^\dagger)X_\alpha$, completion of the Pauli--Fierz square gives
\begin{equation}
\lambda_\alpha^2=\frac{2\hbar g_\alpha^2}{\omega_\alpha},
\label{smeq:lambdag}
\end{equation}
where $\omega_\alpha$ is an angular frequency in Eq.~\ref{smeq:lambdag}.  This relation connects the variational coupling convention directly to standard macroscopic-QED spectral densities.

\section{Planar 4H-SiC surface-phonon-polariton geometry}
\label{smsec:sphp}

\subsection{Macroscopic-QED mapping}

For a local dipole $\mathbf p$ at positions $\mathbf R_i$, the electromagnetic coupling spectral-density matrix is~\cite{Dung1998,Svendsen2021}
\begin{equation}
J^{(g)}_{ij}(\omega)= \frac{\omega^2}{\pi\hbar\epsilon_0c^2}
\mathbf p\cdot\operatorname{Im} \mathbf G(\mathbf R_i,\mathbf R_j;\omega)
\cdot\mathbf p,
\label{smeq:mqedJ}
\end{equation}
where $\mathbf G$ is the classical dyadic Green tensor.  Combining Eqs.~\ref{smeq:lambdag} and \ref{smeq:mqedJ}, the Pauli--Fierz self-polarization weight per angular-frequency interval is
\begin{equation}
\frac{d\Lambda_{ij}}{d\omega}=
\frac{2\omega}{\pi\epsilon_0c^2}
\mathbf p\cdot\operatorname{Im} \mathbf G(\mathbf R_i,\mathbf R_j;\omega)
\cdot\mathbf p.
\label{smeq:LambdaG}
\end{equation}
Only the scattering part of $\mathbf G$ is retained; the free-space contribution already absorbed into the observable matter parameters is subtracted.

Consider a layer in vacuum at height $d$ above a planar polar dielectric and take $\mathbf p=p_z\hat{\mathbf z}$.  In the nonretarded near field, the scattering Green tensor is
\begin{equation}
\operatorname{Im} G_{zz}^{\rm sc}(\mathbf R_i,\mathbf R_j;\omega)
=\frac{1}{4\pi k_0^2}\int_0^\infty dq q^2e^{-2qd}
J_0(qR_{ij})\operatorname{Im}\Rp(q,\omega),
\label{smeq:Gzz}
\end{equation}
with $k_0=\omega/c$. Here $\Rp(q,\omega)$ is the Fresnel reflection coefficient for
$p$-polarized modes with in-plane wave vector $q$ and frequency $\omega$,
\begin{equation}
\Rp(q,\omega)
= \frac{\epsilon(\omega)\beta_0-\beta_m}
{\epsilon(\omega)\beta_0+\beta_m},
\quad
\beta_0=\sqrt{k_0^2-q^2},
\quad
\beta_m=\sqrt{\epsilon(\omega)k_0^2-q^2}.
\end{equation}
In the nonretarded limit $q\gg k_0$, this reduces to
\begin{equation}
\Rp(\omega) = \frac{\epsilon(\omega)-1}
{\epsilon(\omega)+1}.
\end{equation}

Using photon energy $\Omega=\hbar\omega$, the local Pauli--Fierz spectral weight resolved in $q$ is
\begin{equation}
\mathcal J_{\rm PF}^{\rm surf}(\Omega,q)=
\frac{p_z^2}{2\pi^2\epsilon_0\Omega}
q^2e^{-2qd}\operatorname{Im}\Rp(q,\Omega).
\label{smeq:Jsurface}
\end{equation}
For an isotropic nearest-neighbor bond of length $a$, angular averaging gives
\begin{equation}
\rb(q)=1-J_0(qa).
\label{smeq:rq}
\end{equation}
The geometry-specific Mott-boundary shift is therefore
\begin{equation}
\Delta U_c^{\rm surf}=-\int_0^\infty d\Omega\int_0^{q_{\rm max}}dq
\mathcal J_{\rm PF}^{\rm surf}(\Omega,q)
\frac{4A_0\rb(q)}{\Omega+4A_0\rb(q)}.
\label{smeq:surfaceShift}
\end{equation}
We impose $q_{\rm max}=\pi/a$ so that the continuum Green tensor is not extrapolated beyond the electronic Brillouin-zone scale.

\subsection{Dielectric response of 4H-SiC}
The isotropic single-oscillator approximation is
\begin{equation}
\epsilon(\omega)=\epsilon_\infty
\frac{\omega_{\rm LO}^2-\omega^2-i\gamma\omega}
{\omega_{\rm TO}^2-\omega^2-i\gamma\omega},
\label{smeq:epsSiC}
\end{equation}
with
\begin{equation}
\epsilon_\infty=6.7,\quad
\omega_{\rm TO}=797~\mathrm{cm}^{-1},\quad
\omega_{\rm LO}=969~\mathrm{cm}^{-1},\quad
\gamma=4~\mathrm{cm}^{-1},
\end{equation}
based on infrared ellipsometry of 4H-SiC~\cite{Tiwald1999}. The peak of $\operatorname{Im}\Rp$ lies inside the Reststrahlen band and produces the spectrum shown in Fig.~\ref{fig:sphp}(a) of the main text.

\subsection{Distance asymptotics}

When $d\gg a$, the dominant momenta satisfy $q\sim d^{-1}$ and
\begin{equation}
\rb(q)=1-J_0(qa)=\frac{q^2a^2}{4}+O(q^4a^4).
\end{equation}
If additionally $4A_0\rb(q)\ll\Omega$, the kernel in Eq.~\ref{smeq:surfaceShift} is $4A_0\rb(q)/\Omega$.  Using
\begin{equation}
\int_0^\infty dq q^4e^{-2qd}=\frac{3}{4d^5},
\end{equation}
we obtain
\begin{equation}
|\Delta U_c|\simeq
\frac{3p_z^2A_0a^2}{8\pi^2\epsilon_0d^5}
\int_0^\infty d\Omega
\frac{\operatorname{Im}\Rp(\Omega)}{\Omega^2}.
\label{smeq:d5}
\end{equation}
The additional factor $d^{-2}$ relative to the local near-field density of states is the direct signature of the electronic bond-gradient filter.
For $d\sim a$, before the Brillouin-zone cutoff and microscopic nonlocality dominate, the relevant momenta satisfy $qa=O(1)$, $\rb(q)$ becomes order unity, and the kernel approaches a constant.  Since
\begin{equation}
\int_0^\infty dq q^2e^{-2qd}=\frac{1}{4d^3},
\end{equation}
the shift crosses toward the ordinary local $d^{-3}$ near-field scaling over the intermediate atomic-confinement regime.

\subsection{Numerical values}

\begin{table}[h]
\caption{Surface-polariton contribution to the Mott-boundary shift. Here $\Lambda_{\rm loc}$ is the total local Pauli--Fierz self-polarization weight before projection onto the electronic bond form factor.}
\label{smtab:tab_surface}
\begin{ruledtabular}
\begin{tabular}{cccc}
$d$ (nm) & $|\Delta U_c|$ (meV) & $\Lambda_{\rm loc}$ (meV) & $|\Delta U_c|/\Lambda_{\rm loc}$\\
0.5 & 5.04 & 8.54 & 0.590\\
1.0 & 0.372 & 1.084 & 0.343\\
2.0 & 0.0194 & 0.135 & 0.144\\
5.0 & $2.63\times10^{-4}$ & $8.67\times10^{-3}$ & 0.0303
\end{tabular}
\end{ruledtabular}
\end{table}
For Fig.~\ref{fig:sphp} of the main text we use $A_0=0.2$ eV, $a=0.4$ nm, and $p_z=2e\mathring{\mathrm A}$.
Direct integration of Eq.~\ref{smeq:surfaceShift} gives Table~\ref{smtab:tab_surface}.
Here $\Lambda_{\rm loc}=\int d\Omega dq \mathcal J_{\rm PF}^{\rm surf}$ is the total local self-polarization weight.  The rapidly decreasing ratio in the final column shows that a large local photonic density of states is insufficient at large distance: the field must also carry momentum capable of resolving an electronic bond.  Because the shift scales as $p_z^2$, results for other local charge-transfer dipoles follow by direct rescaling.  At $d\lesssim a$, microscopic nonlocality and surface chemistry are not represented by the continuum model, so the first row should be interpreted as an upper-scale estimate rather than a quantitative materials prediction.

\section{Finite-coordination determinant-VMC validation}
\label{smsec:vmc}

The analytical derivation relies on the exact infinite-coordination evaluation of the chosen Gutzwiller variational state. To test the critical coefficient and mode-extensivity scaling without using that contraction, we evaluate the same variational-displacement family by determinant variational Monte Carlo (VMC) on finite random regular graphs. No Gutzwiller operator equivalence is imposed in this calculation. Fermionic Gutzwiller states can be sampled directly by VMC~\cite{HongHirsch1990,Foulkes2001}, while the photon matrix elements are integrated analytically.

\subsection{Finite-graph trial state}

For a $z$-regular graph with $N$ sites, the one-body hopping is $-t_*/\sqrt z$ on every edge.  The reference determinant $\ket{\Phi_0}$ fills the lowest $N/2$ one-particle orbitals for each spin.  We use
\begin{equation}
\ket{\Psi(g,\xi,\sigma)}=
\hat U_\xi g^{\hat D_{\rm tot}}\ket{\Phi_0}\ket{\chi_\sigma},
\quad
\hat D_{\rm tot}=\sum_i n_{i\uparrow}n_{i\downarrow},
\label{smeq:vmcstate}
\end{equation}
with $0<g\leq1$ and the variational displacement in Eq.~\ref{smeq:U}.  For an electronic configuration
$C=(C_\uparrow,C_\downarrow)$ at fixed $N_\uparrow=N_\downarrow=N/2$, the sampled electronic amplitude is
\begin{equation}
\psi_g(C)=g^{D_{\rm tot}(C)} \det[\Phi_0(C_\uparrow)] \det[\Phi_0(C_\downarrow)].
\label{smeq:amplitude}
\end{equation}
Because $\hat U_\xi$ is unitary and acts as a photon translation conditioned on the electronic occupation, the marginal electronic sampling probability is $|\psi_g(C)|^2$ and is independent of $\xi$ and $\sigma$. We denote the measured double occupancy per site by $D=\avg{\hat D_{\rm tot}}/N$.

\subsection{Photon-integrated local-energy estimator}

Let $C_{i\rightarrow j,\sigma}$ denote a configuration obtained by moving a spin-$\sigma$ electron from occupied site $i$ to an empty site $j$, and define the determinant-Gutzwiller ratio
\begin{equation}
R_{ij\sigma}(C)=
\frac{\psi_g(C_{i\rightarrow j,\sigma})}{\psi_g(C)}.
\label{smeq:ratio}
\end{equation}
The Gaussian photon overlap associated with this hop is
\begin{equation}
B_{ij}(\xi,\sigma)= \exp\left[-\frac{\lambda^2\xi^2\sigma}{4\omega}
\sum_\alpha|u_{\alpha i}-u_{\alpha j}|^2\right].
\label{smeq:Bij}
\end{equation}
After integrating the photon coordinates analytically, the local-energy estimator for this variational state is
\begin{align}
E_{\rm loc}(C)=&-\frac{t_*}{\sqrt z}
\sum_{\langle ij\rangle\sigma}^{\rm allowed}
B_{ij}(\xi,\sigma)R_{ij\sigma}(C)
+U D_{\rm tot}(C)\nonumber\\
&+\frac{\lambda^2(1-\xi)^2}{2}
\sum_\alpha X_\alpha(C)^2
+M\frac{\omega}{4}(\sigma+\sigma^{-1}-2).
\label{smeq:Eloc}
\end{align}
The directed hopping sum contains every allowed matrix element of the Hermitian kinetic operator.  Equation~\ref{smeq:Eloc} is evaluated by ordinary Metropolis sampling of Eq.~\ref{smeq:amplitude}.  Determinant ratios and accepted updates are computed with Sherman--Morrison row replacements.  In addition to single-electron moves, spin-exchange proposals are included to maintain efficient sampling as $g$ becomes small.

\subsection{Localized-mode estimator and weak-coupling slope}

For the mode-extensivity benchmark we choose orthonormal site-local modes on a selected set $\mathcal S$,
\begin{equation}
u_{\alpha i}=\delta_{i,i_\alpha},
\quad i_\alpha\in\mathcal S,
\quad M=|\mathcal S|.
\label{smeq:localmodes}
\end{equation}
On any regular graph this ensemble has $\nmode=M/N$ and $r=1$ exactly.  For a hop across bond $ij$, define
\begin{equation}
s_{ij}=\mathbf 1_{\mathcal S}(i)+\mathbf 1_{\mathcal S}(j)\in\{0,1,2\}.
\end{equation}
The hopping estimator can then be accumulated in three classes $K_s$ according to $s_{ij}=s$, and
\begin{equation}
E(g,\xi,\sigma)=
\sum_{s=0}^{2}e^{-\alpha s}\avg{K_s}
+U\avg{\hat D_{\rm tot}}
+\frac{\lambda^2(1-\xi)^2}{2}\avg Q
+M\frac{\omega}{4}(\sigma+\sigma^{-1}-2),
\label{smeq:vmcenergy}
\end{equation}
where
\begin{equation}
\alpha=\frac{\lambda^2\xi^2\sigma}{4\omega},
\quad
Q=\sum_{i\in\mathcal S}(n_i-1)^2.
\end{equation}
Thus a single electronic Markov chain at fixed $g$ supplies all estimators needed to evaluate the energy for arbitrary $\xi$, $\sigma$, $\lambda$, and $\omega$.

At weak coupling, $\sigma=1+O(\lambda^2)$ and the squeezing energy first contributes at $O(\lambda^4)$.  Expanding Eq.~\ref{smeq:vmcenergy} gives
\begin{equation}
\frac{\Delta E_{\rm cav}}{\nmode\lambda^2}
=a\xi^2+b(1-\xi)^2+O(\lambda^2),
\label{smeq:weakvmc}
\end{equation}
with directly sampled coefficients
\begin{align}
a&=-\frac{1}{4\omega N\nmode}
\sum_{s=0}^{2}s\avg{K_s},
\label{smeq:aVMC}\\
b&=\frac{\avg Q}{2N\nmode}.
\label{smeq:bVMC}
\end{align}
The optimal displacement and minimum correction are therefore
\begin{equation}
\xi_0^{\rm VMC}=\frac{b}{a+b},
\quad
\frac{\Delta E_{\rm cav}}{\nmode\lambda^2}
=\frac{ab}{a+b}+O(\lambda^2).
\label{smeq:vmcslope}
\end{equation}
This finite-graph relation contains no infinite-$z$ approximation. Linear mode-density scaling follows directly from normalized-mode counting and the sampled electronic observables.

\subsection{Direct extraction of the critical-shift coefficient}

Within the Gutzwiller variational description, the cavity contribution to the critical interaction is the derivative of the cavity energy with respect to the per-site double occupancy at the projected endpoint.  Define
\begin{equation}
y(D)=\min_\xi\frac{\Delta E_{\rm cav}(D)}{\nmode\lambda^2}
=\frac{a(D)b(D)}{a(D)+b(D)}.
\label{smeq:yD}
\end{equation}
Then
\begin{equation}
C_{\rm VMC}=\left.\frac{dy}{dD}\right|_{D\rightarrow0},
\quad
\Delta U_c=-\nmode\lambda^2C_{\rm VMC}.
\label{smeq:Cvmc}
\end{equation}
For every graph and mode set, the sampled values are fitted over $D\leq0.135$ to
\begin{equation}
y(D)=C_{\rm VMC}D+c_2D^2,
\label{smeq:fitD}
\end{equation}
with the zero intercept fixed by the fully projected state.  The infinite-coordination prediction for site-local modes ($r=1$) is
\begin{equation}
C_{\rm GA}=\frac{4A_0}{\omega+4A_0},
\label{smeq:CGA}
\end{equation}
where $A_0$ is evaluated separately for each finite graph.  This procedure directly tests the coefficient controlling the Mott-boundary shift rather than the cavity energy at a single representative value of $D$.

\subsection{Simulation parameters and statistical analysis}

The critical-shift data in Fig.~\ref{fig:vmcnew}(a,b) of the main text use $N=24$ random regular graphs with $z=4,8,12,16$, 50 independent graph realizations for each coordination, $t_*=1$, and $\omega=2$.  The sampled Gutzwiller parameters are
\begin{equation}
g\in\{0.08,0.14,0.22,0.32,0.48\},
\end{equation}
which span measured double occupancies from approximately $0.02$ to $0.16$.  Each chain is equilibrated for 180 sweeps, followed by 280 measurements separated by two sweeps.  Blocking estimates quantify the Monte Carlo error for each point; the error bars on the extracted $C_{\rm VMC}$ additionally include graph-to-graph variation.  The resulting coefficients are
\begin{table}[t]
\caption{Finite-coordination VMC estimate of the critical-shift coefficient. Parentheses denote statistical uncertainties in the final digits and include graph-to-graph variation.}
\label{smtab:vmc_coeff}
\begin{ruledtabular}
\begin{tabular}{cccc}
$z$ & $C_{\rm VMC}$ & $C_{\rm GA}$ & $C_{\rm VMC}/C_{\rm GA}$\\
4  & $0.6739(18)$ & $0.6317$ & $1.068(2)$\\
8  & $0.6438(13)$ & $0.6142$ & $1.049(2)$\\
12 & $0.6204(13)$ & $0.5921$ & $1.048(2)$ \\
16 & $0.5907(13)$ & $0.5588$ & $1.057(2)$
\end{tabular}
\end{ruledtabular}
\end{table}
The finite-coordination estimate therefore agrees with the infinite-coordination result within $5$--$7\%$ over the range studied. The remaining difference is consistent with finite-$z$ corrections, including nonlocal loop contractions absent from the Brinkman--Rice evaluation.

Figure~\ref{fig:vmcnew}(c) of the main text shows the independent mode-extensivity calculation.  Those data use $N=16$, $z=8$, $g=0.5$, $\omega=2$, and a finite test coupling $\lambda^2=0.2$ with $M=1,2,4,8,16$.  A second sequence fixes $M=1$ and varies $N=10,12,16,20$.  The photon variables are optimized independently for every $M$, $N$, and statistical block.

%% file: references.bib
@article{dmft:2006rmp,
	author = {Kotliar, G. and Savrasov, S. Y. and Haule, K. and Oudovenko, V. S. and Parcollet, O. and Marianetti, C. A.},
	doi = {10.1103/RevModPhys.78.865},
	issue = {3},
	journal = {Rev. Mod. Phys.},
	month = {Aug},
	numpages = {0},
	pages = {865--951},
	publisher = {American Physical Society},
	title = {Electronic structure calculations with dynamical mean-field theory},
	url = {https://link.aps.org/doi/10.1103/RevModPhys.78.865},
	volume = {78},
	year = {2006}}

@article{LanataYao2014,
	author = {Lanat\`a, Nicola and Strand, Hugo U. R. and Yao, Yongxin and Kotliar, Gabriel},
	doi = {10.1103/PhysRevLett.113.036402},
	issue = {3},
	journal = {Phys. Rev. Lett.},
	month = {Jul},
	numpages = {5},
	pages = {036402},
	publisher = {American Physical Society},
	title = {Principle of Maximum Entanglement Entropy and Local Physics of Strongly Correlated Materials},
	url = {https://link.aps.org/doi/10.1103/PhysRevLett.113.036402},
	volume = {113},
	year = {2014}}

@article{Chern2017,
	author = {Chern, Gia-Wei and Barros, Kipton and Batista, Cristian D. and Kress, Joel D. and Kotliar, Gabriel},
	doi = {10.1103/PhysRevLett.118.226401},
	issue = {22},
	journal = {Phys. Rev. Lett.},
	month = {Jun},
	numpages = {5},
	pages = {226401},
	publisher = {American Physical Society},
	title = {Mott Transition in a Metallic Liquid: Gutzwiller Molecular Dynamics Simulations},
	url = {https://link.aps.org/doi/10.1103/PhysRevLett.118.226401},
	volume = {118},
	year = {2017}}

@misc{bretscher2026,
	archiveprefix = {arXiv},
	author = {Hope M Bretscher and Lorenzo Graziotto and Marios H Michael and Angela Montanaro and I-Te Lu and Andrey Grankin and James W McIver and Jerome Faist and Daniele Fausti and Martin Eckstein and Michael Ruggenthaler and Angel Rubio and DN Basov and Mohammad Hafezi and Martin Claassen and Dante M Kennes and Michael A Sentef},
	eprint = {2604.08666},
	primaryclass = {cond-mat.mes-hall},
	title = {Fluctuation engineering in cavity quantum materials},
	url = {https://arxiv.org/abs/2604.08666},
	year = {2026}}

@article{Keren:2026aa,
	author = {Keren, Itai and Webb, Tatiana A. and Zhang, Shuai and Xu, Jikai and Sun, Dihao and Kim, Brian S. Y. and Shin, Dongbin and Zhang, Songtian S. and Zhang, Junhe and Pereira, Giancarlo and Yao, Juntao and Okugawa, Takuya and Michael, Marios H. and Vi{\~n}as Bostr{\"o}m, Emil and Edgar, James H. and Wolf, Stuart and Julian, Matthew and Prasankumar, Rohit P. and Miyagawa, Kazuya and Kanoda, Kazushi and Gu, Genda and Cothrine, Matthew and Mandrus, David and Buzzi, Michele and Cavalleri, Andrea and Dean, Cory R. and Kennes, Dante M. and Millis, Andrew J. and Li, Qiang and Sentef, Michael A. and Rubio, Angel and Pasupathy, Abhay N. and Basov, D. N.},
	date = {2026/02/01},
	doi = {10.1038/s41586-025-10062-6},
	id = {Keren2026},
	isbn = {1476-4687},
	journal = {Nature},
	number = {8103},
	pages = {864--868},
	title = {Cavity-altered superconductivity},
	url = {https://doi.org/10.1038/s41586-025-10062-6},
	volume = {650},
	year = {2026}}

@article{Philbin2010,
	author = {Philbin, T G},
	date = {2010/12/07},
	doi = {10.1088/1367-2630/12/12/123008},
	isbn = {1367-2630;},
	journal = {New Journal of Physics},
	number = {12},
	pages = {123008},
	title = {Canonical quantization of macroscopic electromagnetism},
	url = {https://doi.org/10.1088/1367-2630/12/12/123008},
	volume = {12},
	year = {2010}}

@article{Svendsen2021,
	author = {Svendsen, Mark Kamper and Kurman, Yaniv and Schmidt, Peter and Koppens, Frank and Kaminer, Ido and Thygesen, Kristian S.},
	date = {2021/05/13},
	doi = {10.1038/s41467-021-23012-3},
	id = {Svendsen2021},
	isbn = {2041-1723},
	journal = {Nature Communications},
	number = {1},
	pages = {2778},
	title = {Combining density functional theory with macroscopic QED for quantum light-matter interactions in 2D materials},
	url = {https://doi.org/10.1038/s41467-021-23012-3},
	volume = {12},
	year = {2021}}

@article{Greffet2002,
	author = {Greffet, Jean-Jacques and Carminati, R{\'e}mi and Joulain, Karl and Mulet, Jean-Philippe and Mainguy, St{\'e}phane and Chen, Yong},
	date = {2002/03/01},
	doi = {10.1038/416061a},
	id = {Greffet2002},
	isbn = {1476-4687},
	journal = {Nature},
	number = {6876},
	pages = {61--64},
	title = {Coherent emission of light by thermal sources},
	url = {https://doi.org/10.1038/416061a},
	volume = {416},
	year = {2002}}

@article{Taubner2006,
	author = {Taubner, Thomas and Korobkin, Dmitriy and Urzhumov, Yaroslav and Shvets, Gennady and Hillenbrand, Rainer},
	date = {2006/09/15},
	doi = {10.1126/science.1131025},
	journal = {Science},
	journal1 = {Science},
	journal2 = {Science},
	month = {2026/07/11},
	number = {5793},
	pages = {1595--1595},
	publisher = {American Association for the Advancement of Science},
	title = {Near-Field Microscopy Through a SiC Superlens},
	type = {doi: 10.1126/science.1131025},
	url = {https://doi.org/10.1126/science.1131025},
	volume = {313},
	year = {2006},
	year1 = {2006}}

@article{Caldwell2015,
	author = {Caldwell, Joshua D. and Lindsay, Lucas and Giannini, Vincenzo and Vurgaftman, Igor and Reinecke, Thomas L. and Maier, Stefan A. and Glembocki, Orest J.},
	doi = {doi:10.1515/nanoph-2014-0003},
	journal = {Nanophotonics},
	month = {2026-07-11},
	number = {1},
	pages = {44--68},
	title = {Low-loss, infrared and terahertz nanophotonics using surface phonon polaritons},
	url = {https://doi.org/10.1515/nanoph-2014-0003},
	volume = {4},
	year = {2015}}

@article{Hillenbrand2002,
	author = {Hillenbrand, R. and Taubner, T. and Keilmann, F.},
	date = {2002/07/01},
	doi = {10.1038/nature00899},
	id = {Hillenbrand2002},
	isbn = {1476-4687},
	journal = {Nature},
	number = {6894},
	pages = {159--162},
	title = {Phonon-enhanced light--matter interaction at the nanometre scale},
	url = {https://doi.org/10.1038/nature00899},
	volume = {418},
	year = {2002}}

@article{Schlawin2022,
	author = {Schlawin, F. and Kennes, D. M. and Sentef, M. A.},
	doi = {10.1063/5.0083825},
	isbn = {1931-9401},
	journal = {Applied Physics Reviews},
	journal1 = {Appl. Phys. Rev.},
	month = {7/11/2026},
	number = {1},
	pages = {011312},
	title = {Cavity quantum materials},
	url = {https://doi.org/10.1063/5.0083825},
	volume = {9},
	year = {2022},
	year1 = {2022/02/25}}

@article{Hubener2021,
	author = {H{\"u}bener, Hannes and De Giovannini, Umberto and Sch{\"a}fer, Christian and Andberger, Johan and Ruggenthaler, Michael and Faist, Jerome and Rubio, Angel},
	date = {2021/04/01},
	doi = {10.1038/s41563-020-00801-7},
	id = {H{\"u}bener2021},
	isbn = {1476-4660},
	journal = {Nature Materials},
	number = {4},
	pages = {438--442},
	title = {Engineering quantum materials with chiral optical cavities},
	url = {https://doi.org/10.1038/s41563-020-00801-7},
	volume = {20},
	year = {2021}}

@article{Weight:2023aa,
	author = {Weight, Braden M. and Li, Xinyang and Zhang, Yu},
	doi = {10.1039/d3cp01415k},
	isbn = {1463-9076},
	journal = {Physical Chemistry Chemical Physics},
	journal1 = {Phys. Chem. Chem. Phys.},
	month = {7/10/2026},
	number = {46},
	pages = {31554--31577},
	title = {Theory and modeling of light-matter interactions in chemistry: current and future},
	url = {https://doi.org/10.1039/d3cp01415k},
	volume = {25},
	year = {2023},
	year1 = {2023/12/14}}

@misc{mazin2024,
	archiveprefix = {arXiv},
	author = {Ilia Mazin and Yu Zhang},
	eprint = {2411.15022},
	primaryclass = {quant-ph},
	title = {Light-Matter Hybridization and Entanglement from the First-Principles},
	url = {https://arxiv.org/abs/2411.15022},
	year = {2024}}

@article{Bauman:2025aa,
	author = {Bauman, Nicholas and Cunha, Leonardo A. and DePrince, A. Eugene III and Flick, Johannes and Foley, Jonathan J. IV and Govind, Niranjan and Groenhof, Gerrit and Hoffmann, Norah and Kowalski, Karol and Li, Xiaosong and Liebenthal, Marcus and Maitra, Neepa T. and Manderna, Ruby and Matou{\v s}ek, Mikul{\'a}{\v s} and Mazin, Ilia M. and Mejia-Rodriguez, Daniel and Panyala, Ajay and Peng, Bo and Peyton, Benjamin and Veis, Libor and Vu, Nam and Weidman, Jared D. and Wilson, Angela K. and Zarotiadis, Rhiannon A. and Zhang, Yu},
	date = {2025/10/28},
	doi = {10.1021/acs.jctc.5c00801},
	isbn = {1549-9618},
	journal = {Journal of Chemical Theory and Computation},
	journal1 = {Journal of Chemical Theory and Computation},
	journal2 = {J. Chem. Theory Comput.},
	month = {10},
	number = {20},
	pages = {10035--10067},
	publisher = {American Chemical Society},
	title = {Perspective on Many-Body Methods for Molecular Polaritonic Systems},
	type = {doi: 10.1021/acs.jctc.5c00801},
	url = {https://doi.org/10.1021/acs.jctc.5c00801},
	volume = {21},
	year = {2025},
	year1 = {2025}}

@misc{li2023,
	archiveprefix = {arXiv},
	author = {Xinyang Li and Yu Zhang},
	eprint = {2310.18228},
	primaryclass = {physics.chem-ph},
	title = {First-principles molecular quantum electrodynamics theory at all coupling strengths},
	url = {https://arxiv.org/abs/2310.18228},
	year = {2023}}

@article{Mott1968,
	author = {MOTT, N. F.},
	doi = {10.1103/RevModPhys.40.677},
	issue = {4},
	journal = {Rev. Mod. Phys.},
	month = {Oct},
	numpages = {0},
	pages = {677--683},
	publisher = {American Physical Society},
	title = {Metal-Insulator Transition},
	url = {https://link.aps.org/doi/10.1103/RevModPhys.40.677},
	volume = {40},
	year = {1968}}

@article{Imada1998,
	author = {Imada, Masatoshi and Fujimori, Atsushi and Tokura, Yoshinori},
	doi = {10.1103/RevModPhys.70.1039},
	issue = {4},
	journal = {Rev. Mod. Phys.},
	month = {Oct},
	numpages = {0},
	pages = {1039--1263},
	publisher = {American Physical Society},
	title = {Metal-insulator transitions},
	url = {https://link.aps.org/doi/10.1103/RevModPhys.70.1039},
	volume = {70},
	year = {1998}}

@article{Hubbard1963,
	author = {Hubbard, J.},
	journal = {Proc. R. Soc. A},
	pages = {238--257},
	title = {Electron correlations in narrow energy bands},
	volume = {276},
	year = {1963}}

@article{Gutzwiller1963,
	author = {Gutzwiller, Martin C.},
	doi = {10.1103/PhysRevLett.10.159},
	issue = {5},
	journal = {Phys. Rev. Lett.},
	month = {Mar},
	numpages = {0},
	pages = {159--162},
	publisher = {American Physical Society},
	title = {Effect of Correlation on the Ferromagnetism of Transition Metals},
	url = {https://link.aps.org/doi/10.1103/PhysRevLett.10.159},
	volume = {10},
	year = {1963}}

@article{Gutzwiller1964,
	author = {Gutzwiller, Martin C.},
	doi = {10.1103/PhysRev.134.A923},
	issue = {4A},
	journal = {Phys. Rev.},
	month = {May},
	numpages = {0},
	pages = {A923--A941},
	publisher = {American Physical Society},
	title = {Effect of Correlation on the Ferromagnetism of Transition Metals},
	url = {https://link.aps.org/doi/10.1103/PhysRev.134.A923},
	volume = {134},
	year = {1964}}

@article{Gutzwiller1965,
	author = {Gutzwiller, Martin C.},
	doi = {10.1103/PhysRev.137.A1726},
	issue = {6A},
	journal = {Phys. Rev.},
	month = {Mar},
	numpages = {0},
	pages = {A1726--A1735},
	publisher = {American Physical Society},
	title = {Correlation of Electrons in a Narrow $s$ Band},
	url = {https://link.aps.org/doi/10.1103/PhysRev.137.A1726},
	volume = {137},
	year = {1965}}

@article{BrinkmanRice1970,
	author = {Brinkman, W. F. and Rice, T. M.},
	doi = {10.1103/PhysRevB.2.4302},
	issue = {10},
	journal = {Phys. Rev. B},
	month = {Nov},
	numpages = {0},
	pages = {4302--4304},
	publisher = {American Physical Society},
	title = {Application of Gutzwiller's Variational Method to the Metal-Insulator Transition},
	url = {https://link.aps.org/doi/10.1103/PhysRevB.2.4302},
	volume = {2},
	year = {1970}}

@article{MetznerVollhardt1989,
	author = {Metzner, Walter and Vollhardt, Dieter},
	doi = {10.1103/PhysRevLett.62.324},
	issue = {3},
	journal = {Phys. Rev. Lett.},
	month = {Jan},
	numpages = {0},
	pages = {324--327},
	publisher = {American Physical Society},
	title = {Correlated Lattice Fermions in $d=\ensuremath{\infty}$ Dimensions},
	url = {https://link.aps.org/doi/10.1103/PhysRevLett.62.324},
	volume = {62},
	year = {1989}}

@article{Gebhard1990,
	author = {Gebhard, Florian},
	doi = {10.1103/PhysRevB.41.9452},
	issue = {13},
	journal = {Phys. Rev. B},
	month = {May},
	numpages = {0},
	pages = {9452--9473},
	publisher = {American Physical Society},
	title = {Gutzwiller correlated wave functions in finite dimensions d: A systematic expansion in 1/d},
	url = {https://link.aps.org/doi/10.1103/PhysRevB.41.9452},
	volume = {41},
	year = {1990}}

@article{Bunemann1998,
	author = {B\"unemann, J. and Weber, W. and Gebhard, F.},
	doi = {10.1103/PhysRevB.57.6896},
	issue = {12},
	journal = {Phys. Rev. B},
	month = {Mar},
	numpages = {0},
	pages = {6896--6916},
	publisher = {American Physical Society},
	title = {Multiband Gutzwiller wave functions for general on-site interactions},
	url = {https://link.aps.org/doi/10.1103/PhysRevB.57.6896},
	volume = {57},
	year = {1998}}

@article{Georges1996,
	author = {Georges, Antoine and Kotliar, Gabriel and Krauth, Werner and Rozenberg, Marcelo J.},
	doi = {10.1103/RevModPhys.68.13},
	issue = {1},
	journal = {Rev. Mod. Phys.},
	month = {Jan},
	numpages = {0},
	pages = {13--125},
	publisher = {American Physical Society},
	title = {Dynamical mean-field theory of strongly correlated fermion systems and the limit of infinite dimensions},
	url = {https://link.aps.org/doi/10.1103/RevModPhys.68.13},
	volume = {68},
	year = {1996}}

@article{Kotliar2006,
	author = {Kotliar, G. and Savrasov, S. Y. and Haule, K. and Oudovenko, V. S. and Parcollet, O. and Marianetti, C. A.},
	doi = {10.1103/RevModPhys.78.865},
	issue = {3},
	journal = {Rev. Mod. Phys.},
	month = {Aug},
	numpages = {0},
	pages = {865--951},
	publisher = {American Physical Society},
	title = {Electronic structure calculations with dynamical mean-field theory},
	url = {https://link.aps.org/doi/10.1103/RevModPhys.78.865},
	volume = {78},
	year = {2006}}

@article{Lechermann2007,
	author = {Lechermann, Frank and Georges, Antoine and Kotliar, Gabriel and Parcollet, Olivier},
	doi = {10.1103/PhysRevB.76.155102},
	issue = {15},
	journal = {Phys. Rev. B},
	month = {Oct},
	numpages = {20},
	pages = {155102},
	publisher = {American Physical Society},
	title = {Rotationally invariant slave-boson formalism and momentum dependence of the quasiparticle weight},
	url = {https://link.aps.org/doi/10.1103/PhysRevB.76.155102},
	volume = {76},
	year = {2007}}

@article{Lanata2015,
	author = {Lanat\`a, Nicola and Deng, Xiaoyu and Kotliar, Gabriel},
	doi = {10.1103/PhysRevB.92.081108},
	issue = {8},
	journal = {Phys. Rev. B},
	month = {Aug},
	numpages = {5},
	pages = {081108(R)},
	publisher = {American Physical Society},
	title = {Finite-temperature Gutzwiller approximation from the time-dependent variational principle},
	url = {https://link.aps.org/doi/10.1103/PhysRevB.92.081108},
	volume = {92},
	year = {2015}}

@article{Goldstein2020,
	author = {Goldstein, Garry and Lanata, Nicola and Kotliar, Gabriel},
	doi = {10.1103/PhysRevB.102.045152},
	issue = {4},
	journal = {Phys. Rev. B},
	month = {Jul},
	numpages = {38},
	pages = {045152},
	publisher = {American Physical Society},
	title = {Extending the Gutzwiller approximation to intersite interactions},
	url = {https://link.aps.org/doi/10.1103/PhysRevB.102.045152},
	volume = {102},
	year = {2020}}

@article{Barone2007,
	author = {Barone, P. and Raimondi, R. and Capone, M. and Castellani, C. and Fabrizio, M.},
	doi = {10.1209/0295-5075/79/47003},
	journal = {Europhysics Letters},
	month = {jul},
	number = {4},
	pages = {47003},
	title = {Extended Gutzwiller wave function for the Hubbard-Holstein model},
	url = {https://doi.org/10.1209/0295-5075/79/47003},
	volume = {79},
	year = {2007}}

@article{Barone2008,
	author = {Barone, P. and Raimondi, R. and Capone, M. and Castellani, C. and Fabrizio, M.},
	doi = {10.1103/PhysRevB.77.235115},
	issue = {23},
	journal = {Phys. Rev. B},
	month = {Jun},
	numpages = {10},
	pages = {235115},
	publisher = {American Physical Society},
	title = {Gutzwiller scheme for electrons and phonons: The half-filled Hubbard-Holstein model},
	url = {https://link.aps.org/doi/10.1103/PhysRevB.77.235115},
	volume = {77},
	year = {2008}}

@article{LangFirsov1963,
	author = {Lang, I. G. and Firsov, Y. A.},
	journal = {Sov. Phys. JETP},
	pages = {1301--1312},
	title = {Kinetic theory of semiconductors with low mobility},
	volume = {16},
	year = {1963}}

@article{HongHirsch1990,
	author = {Hong, X. Q. and Hirsch, J. E.},
	doi = {10.1103/PhysRevB.41.4410},
	issue = {7},
	journal = {Phys. Rev. B},
	month = {Mar},
	numpages = {0},
	pages = {4410--4415},
	publisher = {American Physical Society},
	title = {Study of the accuracy of the Gutzwiller wave function for the two-dimensional Hubbard model},
	url = {https://link.aps.org/doi/10.1103/PhysRevB.41.4410},
	volume = {41},
	year = {1990}}

@article{Foulkes2001,
	author = {Foulkes, W. M. C. and Mitas, L. and Needs, R. J. and Rajagopal, G.},
	doi = {10.1103/RevModPhys.73.33},
	issue = {1},
	journal = {Rev. Mod. Phys.},
	month = {Jan},
	numpages = {0},
	pages = {33--83},
	publisher = {American Physical Society},
	title = {Quantum Monte Carlo simulations of solids},
	url = {https://link.aps.org/doi/10.1103/RevModPhys.73.33},
	volume = {73},
	year = {2001}}

@article{Ruggenthaler2014,
	author = {Ruggenthaler, Michael and Flick, Johannes and Pellegrini, Camilla and Appel, Heiko and Tokatly, Ilya V. and Rubio, Angel},
	doi = {10.1103/PhysRevA.90.012508},
	issue = {1},
	journal = {Phys. Rev. A},
	month = {Jul},
	numpages = {26},
	pages = {012508},
	publisher = {American Physical Society},
	title = {Quantum-electrodynamical density-functional theory: Bridging quantum optics and electronic-structure theory},
	url = {https://link.aps.org/doi/10.1103/PhysRevA.90.012508},
	volume = {90},
	year = {2014}}

@article{Flick2017,
	author = {Johannes Flick and Michael Ruggenthaler and Heiko Appel and Angel Rubio},
	doi = {10.1073/pnas.1615509114},
	eprint = {https://www.pnas.org/doi/pdf/10.1073/pnas.1615509114},
	journal = {Proceedings of the National Academy of Sciences},
	number = {12},
	pages = {3026-3034},
	title = {Atoms and molecules in cavities, from weak to strong coupling in quantum-electrodynamics (QED) chemistry},
	url = {https://www.pnas.org/doi/abs/10.1073/pnas.1615509114},
	volume = {114},
	year = {2017}}

@article{Schafer2018,
	author = {Sch\"afer, Christian and Ruggenthaler, Michael and Rubio, Angel},
	doi = {10.1103/PhysRevA.98.043801},
	issue = {4},
	journal = {Phys. Rev. A},
	month = {Oct},
	numpages = {24},
	pages = {043801},
	publisher = {American Physical Society},
	title = {Ab initio nonrelativistic quantum electrodynamics: Bridging quantum chemistry and quantum optics from weak to strong coupling},
	url = {https://link.aps.org/doi/10.1103/PhysRevA.98.043801},
	volume = {98},
	year = {2018}}

@article{Schafer2020,
	author = {Sch{\"a}fer, Christian and Ruggenthaler, Michael and Rokaj, Vasil and Rubio, Angel},
	doi = {10.1021/acsphotonics.9b01649},
	eprint = {https://doi.org/10.1021/acsphotonics.9b01649},
	journal = {ACS Photonics},
	number = {4},
	pages = {975-990},
	title = {Relevance of the Quadratic Diamagnetic and Self-Polarization Terms in Cavity Quantum Electrodynamics},
	url = {https://doi.org/10.1021/acsphotonics.9b01649},
	volume = {7},
	year = {2020}}

@article{Sentef2018,
	author = {Sentef, M. A. and Ruggenthaler, M. and Rubio, A.},
	journal = {Sci. Adv.},
	pages = {eaau6969},
	title = {Cavity quantum-electrodynamical polaritonically enhanced electron-phonon coupling and its influence on superconductivity},
	volume = {4},
	year = {2018}}

@article{Curtis2019,
	author = {Curtis, Jonathan B. and Raines, Zachary M. and Allocca, Andrew A. and Hafezi, Mohammad and Galitski, Victor M.},
	doi = {10.1103/PhysRevLett.122.167002},
	issue = {16},
	journal = {Phys. Rev. Lett.},
	month = {Apr},
	numpages = {5},
	pages = {167002},
	publisher = {American Physical Society},
	title = {Cavity Quantum Eliashberg Enhancement of Superconductivity},
	url = {https://link.aps.org/doi/10.1103/PhysRevLett.122.167002},
	volume = {122},
	year = {2019}}

@article{Kiffner2019,
	author = {Kiffner, Martin and Coulthard, Jonathan R. and Schlawin, Frank and Ardavan, Arzhang and Jaksch, Dieter},
	doi = {10.1103/PhysRevB.99.085116},
	issue = {8},
	journal = {Phys. Rev. B},
	month = {Feb},
	numpages = {9},
	pages = {085116},
	publisher = {American Physical Society},
	title = {Manipulating quantum materials with quantum light},
	url = {https://link.aps.org/doi/10.1103/PhysRevB.99.085116},
	volume = {99},
	year = {2019}}

@article{Mazza2019,
	author = {Mazza, Giacomo and Georges, Antoine},
	doi = {10.1103/PhysRevLett.122.017401},
	issue = {1},
	journal = {Phys. Rev. Lett.},
	month = {Jan},
	numpages = {6},
	pages = {017401},
	publisher = {American Physical Society},
	title = {Superradiant Quantum Materials},
	url = {https://link.aps.org/doi/10.1103/PhysRevLett.122.017401},
	volume = {122},
	year = {2019}}

@article{Andolina2019,
	author = {Andolina, G. M. and Pellegrino, F. M. D. and Giovannetti, V. and MacDonald, A. H. and Polini, M.},
	doi = {10.1103/PhysRevB.100.121109},
	issue = {12},
	journal = {Phys. Rev. B},
	month = {Sep},
	numpages = {6},
	pages = {121109(R)},
	publisher = {American Physical Society},
	title = {Cavity quantum electrodynamics of strongly correlated electron systems: A no-go theorem for photon condensation},
	url = {https://link.aps.org/doi/10.1103/PhysRevB.100.121109},
	volume = {100},
	year = {2019}}

@article{Guerci2020,
	author = {Guerci, Daniele and Simon, Pascal and Mora, Christophe},
	doi = {10.1103/PhysRevLett.125.257604},
	issue = {25},
	journal = {Phys. Rev. Lett.},
	month = {Dec},
	numpages = {6},
	pages = {257604},
	publisher = {American Physical Society},
	title = {Superradiant Phase Transition in Electronic Systems and Emergent Topological Phases},
	url = {https://link.aps.org/doi/10.1103/PhysRevLett.125.257604},
	volume = {125},
	year = {2020}}

@article{Cui2024,
	author = {Cui, Z.-H. and Mandal, A. and Reichman, D. R.},
	journal = {J. Chem. Theory Comput.},
	pages = {1143--1155},
	title = {Variational {Lang--Firsov} approach plus {M{\o}ller--Plesset} perturbation theory with applications to ab initio polariton chemistry},
	volume = {20},
	year = {2024}}

@article{Nakamoto2025,
	author = {Nakamoto, Taiga and Takasan, Kazuaki and Tsuji, Naoto},
	doi = {10.1103/pxkp-trx5},
	issue = {15},
	journal = {Phys. Rev. B},
	month = {Oct},
	numpages = {13},
	pages = {155150},
	publisher = {American Physical Society},
	title = {One-dimensional extended Hubbard model coupled with an optical cavity},
	url = {https://link.aps.org/doi/10.1103/pxkp-trx5},
	volume = {112},
	year = {2025}}

@article{GrunwaldSpectroscopy2025,
	author = {Grunwald, Lukas and Bostr\"om, Emil Vi\~nas and Svendsen, Mark Kamper and Kennes, Dante M. and Rubio, Angel},
	doi = {10.1103/1lpw-22np},
	issue = {24},
	journal = {Phys. Rev. Lett.},
	month = {Jun},
	numpages = {7},
	pages = {246901},
	publisher = {American Physical Society},
	title = {Cavity Spectroscopy for Strongly Correlated Polaritonic Systems},
	url = {https://link.aps.org/doi/10.1103/1lpw-22np},
	volume = {134},
	year = {2025}}

@misc{Inacio2026,
	archiveprefix = {arXiv},
	author = {In{\'a}cio, J. C. and Costa, N. C. and Assaad, F. F.},
	eprint = {2606.03733},
	title = {Mott transition of photons: Quantum {Monte Carlo} study of {Gross--Neveu} criticality in a cavity},
	year = {2026}}

@misc{GrunwaldMultimode2026,
	archiveprefix = {arXiv},
	author = {Grunwald, L. and Cheng, X. and Vi{\~n}as Bostr{\"o}m, E. and Ruggenthaler, M. and Michael, M. H. and Kennes, D. M. and Rubio, A.},
	eprint = {2603.18933},
	title = {Cavity control of strongly correlated electrons beyond resonant coupling},
	year = {2026}}

@article{Dung1998,
	author = {Dung, Ho Trung and Kn\"oll, Ludwig and Welsch, Dirk-Gunnar},
	doi = {10.1103/PhysRevA.57.3931},
	issue = {5},
	journal = {Phys. Rev. A},
	month = {May},
	numpages = {0},
	pages = {3931--3942},
	publisher = {American Physical Society},
	title = {Three-dimensional quantization of the electromagnetic field in dispersive and absorbing inhomogeneous dielectrics},
	url = {https://link.aps.org/doi/10.1103/PhysRevA.57.3931},
	volume = {57},
	year = {1998}}

@article{ScheelBuhmann2008,
	author = {Scheel, S. and Buhmann, S. Y.},
	journal = {Acta Phys. Slovaca},
	pages = {675--809},
	title = {Macroscopic quantum electrodynamics---concepts and applications},
	volume = {58},
	year = {2008}}

@article{Tiwald1999,
	author = {Tiwald, Thomas E. and Woollam, John A. and Zollner, Stefan and Christiansen, Jim and Gregory, R. B. and Wetteroth, T. and Wilson, S. R. and Powell, Adrian R.},
	doi = {10.1103/PhysRevB.60.11464},
	issue = {16},
	journal = {Phys. Rev. B},
	month = {Oct},
	numpages = {0},
	pages = {11464--11474},
	publisher = {American Physical Society},
	title = {Carrier concentration and lattice absorption in bulk and epitaxial silicon carbide determined using infrared ellipsometry},
	url = {https://link.aps.org/doi/10.1103/PhysRevB.60.11464},
	volume = {60},
	year = {1999}}

@article{Eckhardt2025,
	author = {Eckhardt, Christian J. and Grankin, Andrey and Kennes, Dante M. and Ruggenthaler, Michael and Rubio, Angel and Sentef, Michael A. and Hafezi, Mohammad and Michael, Marios H.},
	doi = {10.1103/2fw2-lbhy},
	issue = {15},
	journal = {Phys. Rev. Lett.},
	month = {Oct},
	numpages = {7},
	pages = {156902},
	publisher = {American Physical Society},
	title = {Surface-Mediated Ultrastrong Cavity Coupling of Two-Dimensional Itinerant Electrons},
	url = {https://link.aps.org/doi/10.1103/2fw2-lbhy},
	volume = {135},
	year = {2025}}
